\documentclass[11pt]{article}
\usepackage{amsmath,amssymb,amsfonts,amsthm,epsfig}
\usepackage[usenames,dvipsnames]{xcolor}
\usepackage{bm,xspace}
\usepackage{tcolorbox}
\usepackage{cancel}
\usepackage{fullpage}
\usepackage{liyang}
\usepackage{framed}
\usepackage{verbatim}
\usepackage{enumitem}
\usepackage{array}
\usepackage{multirow}
\usepackage{afterpage}
\usepackage{mathrsfs}
\usepackage{pifont} 
\usepackage{chngpage}
\usepackage[normalem]{ulem}
\usepackage{boxedminipage}
\usepackage{caption}
\usepackage{subcaption}
\usepackage{bold-extra}
\usepackage{forest}
\usepackage{etoolbox}
\usepackage{thm-restate}
\usepackage{microtype}

\usepackage{algorithm}
\usepackage{algorithmicx}
\usepackage{algpseudocode}

\usepackage{tikz}
\usetikzlibrary{calc,through,backgrounds,decorations.pathreplacing, calligraphy,arrows.meta}
\usetikzlibrary{positioning,chains,fit,shapes}
\usetikzlibrary{patterns}

\usepackage{tikz-cd}

\usepackage{pgfplots}
\pgfplotsset{width=8cm,compat=newest}

\def\colorful{1}

\ifnum\colorful=1

\fi
\ifnum\colorful=0

\fi

\newcommand{\D}{\mathcal{D}}
\newcommand{\mcC}{\mathcal{C}}

\newcommand{\mcO}{\mathcal{O}}

\newcommand{\error}{\mathrm{error}}

\newcommand{\NP}{\mathsf{NP}}

\renewcommand{\P}{\mathsf{P}}

\newcommand{\Enc}{\mathrm{Enc}}
\newcommand{\Dec}{\mathrm{Dec}}
\newcommand{\Cert}{\textsc{Cert}}

\newcommand{\UC}{\textsc{UnifCert}}

\newcommand{\VCdim}{\mathrm{VCdim}}
\newcommand{\NTIME}{\mathsf{NTIME}}
\newcommand{\RTIME}{\mathsf{RTIME}}
\newcommand{\AMTIME}{\mathsf{AMTIME}}
\newcommand{\AM}{\mathsf{AM}}

\newcommand{\RP}{\mathsf{RP}}
\newcommand{\SAT}{\mathrm{SAT}}

\newcommand{\Ldim}{\mathrm{Ldim}}

\newtheorem{hypothesis}{Hypothesis}

\Crefname{theorem}{Theorem}{Theorems}
\Crefname{thmenumi}{Theorem}{Theorems}
\AtBeginEnvironment{theorem}{%
    \crefalias{enumi}{thmenumi}%
    \setlist[enumerate,1]{
        label={\textit{(\roman*)}},
        ref={\thetheorem(\roman*)}
    }%
}

\begin{document}


\title{

Computational-Statistical Tradeoffs from NP-hardness \vspace{10pt}


}

\author{ 
Guy Blanc \vspace{6pt} \\ 
\hspace{-7pt} {\sl Stanford} \and 
Caleb Koch \vspace{6pt} \\ 
\hspace{-10pt} { {\sl Stanford}} \and 
 Carmen Strassle \vspace{6pt} \\
 \hspace{-5pt} { {\sl Stanford}} \vspace{15pt}
 \and 
Li-Yang Tan \vspace{6pt}  \\
\hspace{-10pt} {{\sl Stanford}}
}

\date{\small{\today}}

 \maketitle

 \begin{abstract}

 A central question in computer science and statistics is whether efficient algorithms can achieve the information-theoretic limits of statistical problems.
Many  computational-statistical tradeoffs have been shown under average-case assumptions, but since statistical problems are average-case in nature, it has been a challenge to base them on standard worst-case assumptions.

In PAC learning where such tradeoffs were first studied, the question is  whether computational efficiency can come at the cost of using more samples than information-theoretically necessary. We base such tradeoffs  on  $\NP$-hardness and obtain: 


\begin{itemize}
    \item[$\circ$]  
Sharp computational-statistical tradeoffs assuming $\NP$ requires exponential time: For every polynomial $p(n)$, there is an $n$-variate class $\mathcal{C}$ with VC dimension~$1$ such that the sample complexity of {\sl time-efficiently} learning~$\mathcal{C}$ is $\Theta(p(n)).$

\item[$\circ$] A characterization of $\RP$ vs.~$\NP$ in terms of learning: $\RP = \NP$ iff every $\NP$-enumerable class is learnable with $O(\VCdim(\mathcal{C}))$ samples in polynomial time. The forward implication has been known since (Pitt and Valiant, 1988); we prove the reverse implication. 
\end{itemize}

 Notably, all our lower bounds hold  against improper learners. These are the first $\NP$-hardness results for improperly learning a subclass of polynomial-size circuits, circumventing formal barriers of Applebaum, Barak, and Xiao (2008). 
\end{abstract}

 \thispagestyle{empty}
 \newpage


 \setcounter{page}{1}

\section{Introduction}

\label{sec:intro}
This paper is concerned with the relationship between two fundamental resources in learning: runtime and samples. Naturally, we would like algorithms that are both time- {\sl and} sample-efficient. However, such algorithms have proven elusive even for simple learning problems,   suggesting the existence of inherent tradeoffs between the two resources. Are there problems for which any sample-efficient algorithm must be time-inefficient? Are there problems for which more samples provably lead to faster algorithms?

The possibility of such computational-statistical tradeoffs was first raised in an early landmark paper of Blumer, Ehrenfeucht, Haussler, and Warmuth~\cite{BEHW89}. Their paper proved that the Vapnik--Chervonenkis (VC) dimension~\cite{VC71} characterizes the sample complexity of PAC learning, a foundational result now known as the ``Fundamental Theorem of PAC Learning". While undoubtedly important, this characterization is crucially statistical in nature and does not take time complexity into consideration. The upper bound is realized by empirical risk minimization, a naive implementation of which involves iterating over all possible hypotheses. 
The authors themselves noted this and concluded their paper as follows:  
\begin{quote}
   {\sl [C]onsiderable further research remains to be done in this
area. In particular, there may be interesting general trade-offs between the sample
size required for learning and the computational effort required to produce a
consistent hypothesis that are yet to be discovered. These issues are important if this theory of learnability is to find useful applications.} 
\end{quote}

\paragraph{Existing tradeoffs and their  assumptions.} The first such tradeoffs were obtained by Decatur, Goldreich, and Ron~\cite{DGR00} who constructed a concept class~$\mathcal{C}$ with small VC dimension---and is therefore information-theoretically learnable with a small number of samples---and yet is such that, under cryptographic assumptions, a much larger number of samples is both necessary and sufficient to {\sl time-efficiently} learn~$\mathcal{C}$. 

Computational-statistical tradeoffs carry both a negative and a positive message. On one hand, they say that time-efficient learning is impossible even with many more samples than information-theoretically necessary. The size of this gap corresponds to the width of the yellow region in~\Cref{fig:generic tradeoff}. On the other hand, they also say that just a small increase in sample complexity---the transition from the yellow to the green region in~\Cref{fig:generic tradeoff}---leads to substantial speedups in the runtimes of algorithms.
\medskip 

\begin{figure}[h!]
    \centering
    \begin{tikzpicture}
        \def\rectHeight{0.8} 
        \def\rectWidth{5} 

        \fill[pattern=north east lines, pattern color=red!25] 
            (0, -\rectHeight) rectangle (\rectWidth, \rectHeight)
            node[pos=0.5, text width=5cm, text centered] {\normalsize{Information-theoretically \\ impossible}};
        
        \fill[pattern=north east lines, pattern color=yellow] 
            (\rectWidth, -\rectHeight) rectangle (2*\rectWidth, \rectHeight)
            node[pos=0.5, text width=4.5cm, text centered] {\normalsize{Computationally intractable}};


        \fill[pattern=north east lines, pattern color=green!30] 
            (2*\rectWidth, -\rectHeight) rectangle (3*\rectWidth, \rectHeight)
            node[pos=0.5, text width=4cm, text centered] {\normalsize{Computationally tractable}};

        \draw[-Stealth] (0,-1.5) -- (3*\rectWidth,-1.5) node[midway, fill=white]{Sample complexity};
        
        \draw[dashed] (\rectWidth, -\rectHeight-0.2) -- (\rectWidth, \rectHeight+0.2);
        \draw[dashed] (2*\rectWidth, -\rectHeight-0.2) -- (2*\rectWidth, \rectHeight+0.2);
        
    \end{tikzpicture}
    \caption{An illustration of a computational-statistical tradeoff}
    \label{fig:generic tradeoff}
\end{figure}
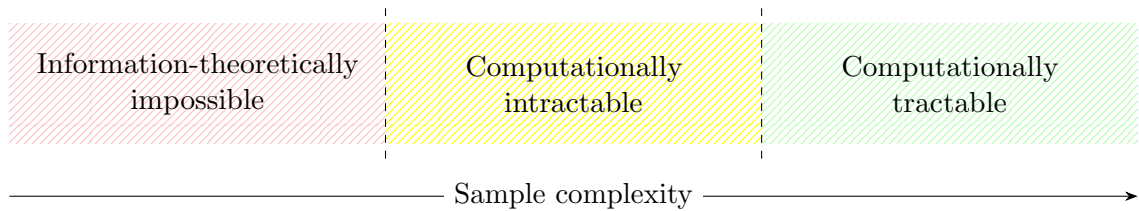

Following~\cite{DGR00}, more such tradeoffs were obtained under cryptographic~\cite{Ser99,SSST12} and more general average-case  assumptions~\cite{DLSS13}.  Beyond PAC learning, computational-statistical tradeoffs have emerged as a recurring phenomenon in areas spanning  computer science, statistics, and statistical physics. We refer the reader to~\cite{JM15,BB20,Sim21} for overviews of the now sizable literature on this topic.

This research has led to a rich theory of average-case reductions and an improved understanding of restricted classes of algorithms. However, connections to worst-case complexity---and specifically $\NP$-hardness, the canonical hardness assumption---have been wanting. The crux of the difficulty has long been recognized: Statistical problems are, by nature, average-case problems, making it challenging to base their hardness on standard worst-case assumptions. 


\paragraph{Is PAC learning $\NP$-hard?} 
 This ties into a major effort in PAC learning that dates back to its inception: basing  the intractability of learning on minimal complexity assumptions. As Valiant already observed in~\cite{Val84},  runtime lower bounds for learning follow straightforwardly from cryptographic assumptions: The very definition of a pseudorandom function~\cite{GGM86} implies that it cannot be learned in polynomial time. 

Such  intractability results are incomparable to computational-statistical tradeoffs. See \Cref{fig:intractable}. On one hand, they are  stronger in the sense of being runtime lower bounds that hold regardless of sample complexity. On the other hand (as perhaps evidenced by the chronological order of~\cite{Val84}'s observation and~\cite{DGR00}'s result), tradeoffs are more subtle in the sense that they can only be shown for more ``delicate" problems that {\sl can} be time-efficiently learned with sufficiently many samples. Put differently, an intractability result involves only a lower bound, whereas a tradeoff involves both a lower and an upper bound.

 \medskip  

\begin{figure}[h!]
    \centering
    \begin{tikzpicture}
        \def\rectHeight{0.8} 
        \def\rectWidth{5} 

        \fill[pattern=north east lines, pattern color=red!25] 
            (0, -\rectHeight) rectangle (\rectWidth, \rectHeight)
            node[pos=0.5, text width=5cm, text centered] {\normalsize{Information-theoretically \\ impossible}};
        
        \fill[pattern=north east lines, pattern color=yellow] 
            (\rectWidth, -\rectHeight) rectangle (3*\rectWidth, \rectHeight)
            node[pos=0.5, text width=5.5cm, text centered] {\normalsize{Computationally intractable}};


        \draw[-Stealth] (0,-1.5) -- (3*\rectWidth,-1.5) node[midway, fill=white]{Sample complexity};
        
        \draw[dashed] (\rectWidth, -\rectHeight-0.2) -- (\rectWidth, \rectHeight+0.2);
    \end{tikzpicture}
    \caption{An illustration of a computationally intractable learning problem}
    \label{fig:intractable}
\end{figure}
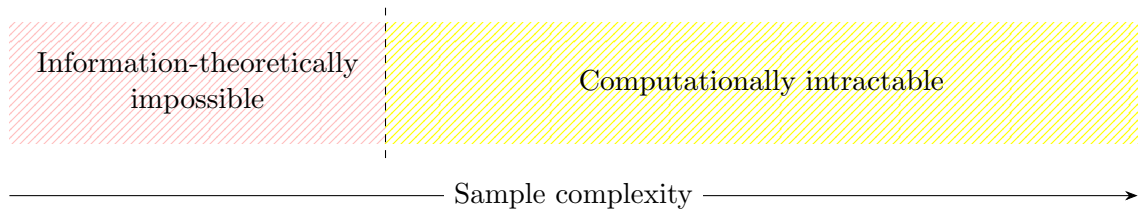

In the same paper, Valiant raised the possibility of basing the intractability of learning on $\NP$-hardness, and subsequently with Pitt~\cite{PV88} they initiated the systematic study of this question.  Despite sustained  efforts, and influential lines of work on the power of cryptographic~\cite{KV94} and more general average-case assumptions~\cite{DLSS14} in learning,  the connection to $\NP$-hardness remains elusive. In fact, Applebaum, Barak, and Xiao~\cite{ABX08} 
established formal barriers to obtaining runtime lower bounds for learning from $\NP$-hardness: They showed that the polynomial hierarchy collapses if such a result is obtained using certain standard types of reductions.


\section{This Work}

 We base computational-statistical tradeoffs in PAC learning on $\NP$-hardness (\Cref{thm:NP intro}), the first such tradeoff across areas based on $\NP$-hardness.

Our lower bound alone carries two implications. First, it shows that the barriers of~\cite{ABX08} can be circumvented: it {\sl is} possible to obtain runtime lower bounds for learning from $\NP$-hardness, via a standard type of reduction, if sample complexity is also taken into account. \Cref{table} summarizes how this fits within the broader context discussed above.\medskip

\backrefsetup{disable}
\begin{table}[h]
 \begin{adjustwidth}{-3.2em}{}
\centering
\renewcommand{\arraystretch}{2}
\begin{tabular}{|c|c|c|}
\hline
 & 
 ~~\begin{tabular}{@{}c@{}}
 Intractability of Learning 
 \vspace{3pt}
       \end{tabular}~
       & 
 \begin{tabular}{@{}c@{}}
  Computational-Statistical Tradeoffs 
  \vspace{3pt}
  \end{tabular}
  \\
\hline \hline 
\begin{tabular}{@{}c@{}}
~~{Cryptographic and}~~ \vspace{-12pt} \\
~~{Average-case assumptions}~~ \end{tabular}  & 
~{\color{teal}$\checkmark$} \cite{Val84,KV94}, \cite{DLSS14},\,...~ & ~{\color{teal}$\checkmark$} \cite{DGR00,Ser99,SSST12}, \cite{DLSS13},\,...~ \\
\hline
\begin{tabular}{@{}c@{}}
$\NP$-hardness
\vspace{3pt}\end{tabular}
& {\color{purple}{$\boldsymbol{\times}$}} Barriers of~\cite{ABX08}& {\color{teal}$\checkmark$} {\bf This work} \\ \hline 
\end{tabular}
\end{adjustwidth}
\caption{Our result in the context of prior work.}
\label{table}
\end{table}
\backrefsetup{enable}

Second, when combined with an observation of~\cite{PV88}, our lower bound yields an {\sl equivalence} between the hardness of learning and the hardness of $\NP$ (\Cref{cor:dichotomy}). 



\subsection{Computational-statistical tradeoffs and the hardness of $\NP$}


Our main result is as follows:\medskip

 \begin{tcolorbox}[colback = white,arc=1mm, boxrule=0.25mm]
\vspace{5pt} 
 \begin{theorem}[Informal; see~\Cref{thm:NP formal} for the formal version]
\label{thm:NP intro}
For every growth function $p(n)\ge n$, there is a concept class $\mathcal{C}$ over $\zo^n$   satisfying the following: 
    \begin{enumerate}
    \item {\bf Very Few Samples $\Rightarrow$ Exists Slow Learner:}  $\VCdim(\mathcal{C}) = 1$ and $\mathcal{C}$ is learnable with $O(1)$ samples in $2^{O(p(n))}$ time.\label{thm:NP intro:VCdim}
            \item {\bf Few Samples $\Rightarrow$ Requires Slow Learners:} If $\NTIME(p(n)) \not\sse \RTIME(t(n))$, any algorithm that learns $\mathcal{C}$ with $O(\log t(n))$ samples requires $(t(n)/p(n))^{\Omega(1)}$  time.\label{thm:NP intro:few samples}
        \item {\bf Many samples $\Rightarrow$ Exists Fast Learner:} There is an algorithm that learns $\mathcal{C}$ with $O(p(n))$ samples in $O(p(n))$ time. \label{thm:NP intro:many samples}
    \end{enumerate}
Furthermore, each concept in $\mcC$ is computed by a circuit of size $O(p(n))$.    
\end{theorem}
\vspace{1pt}
\vspace{-5pt}
\end{tcolorbox}




    
\paragraph{Example settings of parameters.} Assuming $\RP \ne \NP$,
we get a  class $\mathcal{C}$ with $\VCdim(\mathcal{C}) = 1$ that is learnable with $O(n)$ samples in $O(n)$ time, and yet any algorithm using $O(\log n)$ samples must take superpolynomial time.  

Stronger assumptions yield sharper tradeoffs. Under the randomized exponential time hypothesis (ETH), we get a  class $\mathcal{C}$ with $\VCdim(\mathcal{C})= 1$ that is learnable with $O(n)$ samples in $O(n)$ time, and yet there is a constant $\delta > 0$ such that any algorithm using $\le \delta n$ samples must take $2^{\Omega(n)}$ time. This is therefore a class that is information-theoretically learnable from $O(1)$ samples, for which $\Theta(n)$ samples are both necessary and sufficient for {\sl time-efficient} learning.  See~\Cref{fig:ETH}.

\begin{figure}[h!]
    \hspace{1.25em}
    \begin{tikzpicture}
    \def\l{-0.4}
    \def\lp{-0.3}
    \begin{axis}[
            scale only axis,
            width=0.6\linewidth,
            xtick={-0.3,0.75,0.925},
            xticklabels={$O(1)$,{\raisebox{-1.5ex}{$\delta n$}},$O(n)$},
            ytick={0.2,0.8,1},
            yticklabels={$O(n)$,$2^{\Omega(n)}$,$2^{O(n)}$},
            xmin=\l, xmax=1.4,
            ymin=0, ymax=1.1,
            xlabel = {Sample complexity},
            ylabel = {Runtime},
            y label style = {rotate=-90},
            axis lines=left,
            clip=false,
            legend pos=north west,
            legend style={font=\small, nodes={scale=1.25, transform shape}},
            legend image post style={scale=1.5},
        ]
        \addplot[purple, forget plot] coordinates {(0.925,0.2) (1.4,0.2)};
        \addplot[purple, forget plot] coordinates {(\lp,0.8) (0.75,0.8)};
        \addplot[purple, forget plot] coordinates {(\lp,1) (1.4,1)};

        \addplot[gray, dashed, forget plot] coordinates {(\lp,0) (\lp,1)};
        \addplot[gray, dashed, forget plot] coordinates {(0.75,0) (0.75,0.8)};
        \addplot[gray, dashed, forget plot] coordinates {(0.925,0) (0.925,0.2)};

        \addplot[color=blue,smooth,thick,-,domain=\lp:0.651, forget plot] {0.899};
        \addplot[color=blue,smooth,thick,-,domain=0.9:1.4, forget plot] {0.104};
        \addplot[color=blue,smooth,thick,-,domain=0.65:0.9, forget plot] {-0.0045*atan(300*(x-0.8))+0.5};

        \draw[-Stealth] (0.6,1) to (0.6,0.95);
        \draw[-Stealth] (0.23,0.8) to (0.23,0.85);
        \draw[-Stealth] (1.15,0.2) to (1.15,0.15);

        \draw[] (0.6,1.05) node [fill=white] {\small{\cref{thm:NP intro:VCdim}}};
        \draw[] (0.23,0.75) node [fill=white] {\small{\cref{thm:NP intro:few samples}}};
        \draw[] (1.15,0.25) node [fill=white] {\small{\cref{thm:NP intro:many samples}}};
        
    \end{axis}
\end{tikzpicture}
  \captionsetup{width=.93\linewidth}
\caption{Computational-statistical tradeoffs under randomized ETH. Time-efficient learning of this class is impossible even with $\delta n$ samples, many more than information-theoretically necessary ($O(1)$). On the other hand, a constant-factor increase in sample complexity, from $\delta n$ to $O(n)$, leads to an exponential speedup in the runtimes of algorithms.}
\label{fig:ETH}
\end{figure}
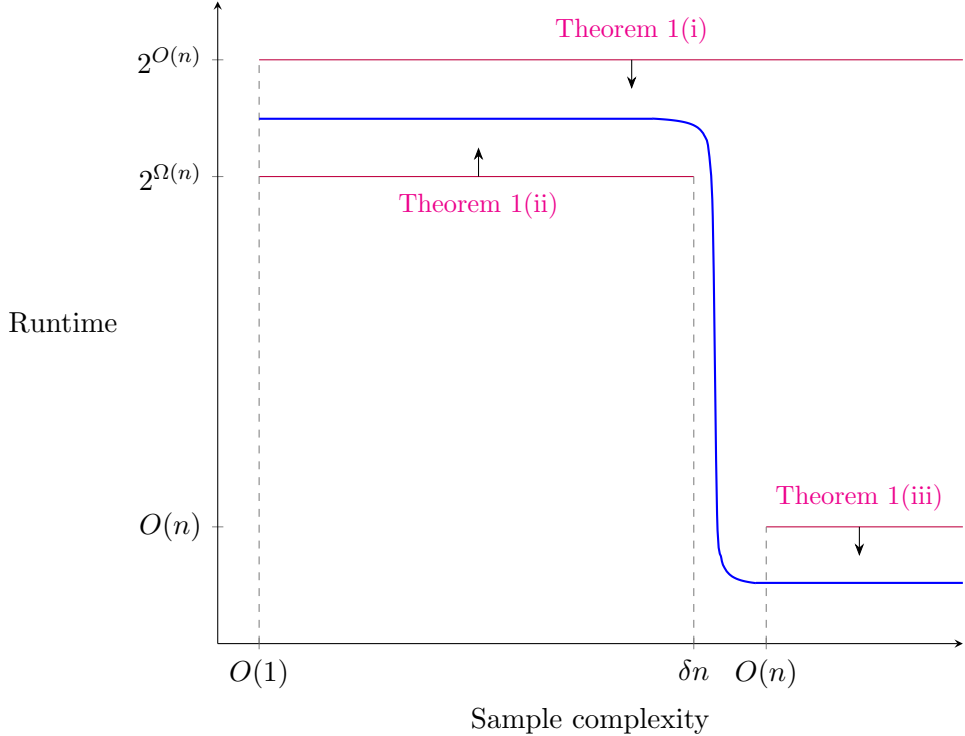

\Cref{thm:NP intro} gives an arbitrarily large separation of $O(1)$ vs.~$\Theta(p(n))$ samples under an appropriate generalization of randomized ETH: There are languages in $\NTIME(p(n))$ that require $2^{\Omega(p(n))}$ time on randomized Turing machines. This is a strong but still standard assumption~\cite{AHU74,Mor81}, saying that the trivial simulation is essentially the best possible.


\subsection{Corollary of~\Cref{thm:NP intro}: An equivalence between the hardness of learning and the hardness of $\NP$}
\label{sec:equivalence}



A basic fact about the computational complexity of learning is that a broad swath of concept classes {\sl can} be time- and sample-efficiently learned if $\RP = \NP$. 
A concept class $\mathcal{C}$ is {\sl $\NP$-enumerable} if there is an efficient algorithm that, given access to an $\NP$ oracle, enumerates all concepts in~$\mathcal{C}$; see~\Cref{def:enumerable-specific} for a formal definition.  Almost all commonly-studied concept classes are $\NP$-enumerable. 
Pitt and Valiant observed the following: 

\begin{fact}[\cite{PV88}]
\label{fact:if RP=NP}
    If $\RP = \NP$, every $\NP$-enumerable class can be learned in polynomial time with $O(\VCdim(\mathcal{C}))$  samples. 
\end{fact}

The proof of~\Cref{fact:if RP=NP} is simple and follows from the fact that empirical risk minimization can be time-efficiently implemented for $\NP$-enumerable classes if $\RP = \NP$. 

As noted in~\cite{PV88}, the significance of this fact is that short of proving $\RP \ne \NP$, lower bounds against $\NP$-enumerable classes have to be conditioned on complexity-theoretic assumptions---ideally on $\RP \ne \NP$, matching~\Cref{fact:if RP=NP}. Their paper initiated research on such lower bounds.


\paragraph{Representation-dependent vs.~-independent hardness.} 
This effort has proven fruitful in the representation-{\sl dependent} setting where the learner is restricted to output a hypothesis of a certain form (e.g.~proper learning). By now, there is a large body of results showing that it is $\NP$-hard to learn a certain concept class $\mathcal{C}$ with a certain hypothesis class $\mathcal{H}$. See the survey~\cite{Fel16} and the textbook~\cite[Chapter 8]{SB14}.

 However, representation-dependent results leave open the possibility that the learning task at hand becomes tractable for learners that are allowed to return more expressive hypotheses.\footnote{In fact, many representation-dependent hardness results are shown for problems that {\sl can} be efficiently learnable in the representation-independent setting. For example, Pitt and Valiant showed, via a reduction from $0$-$1$ integer programming, that the class $\zo$-weight halfspaces is not properly learnable efficiently unless $\RP = \NP$. The restriction to proper learners is clearly crucial here, since even the class of all halfspaces is efficiently learnable.} Much of the power of learning algorithms comes from their ability to return any hypothesis of their choosing, and so representation-{\sl independent} (i.e.~improper learning) hardness results are important for our understanding of their true limitations. Quoting~\cite{KV94}, ``For practical purposes, the efficient learnability of a [concept] class must be considered unresolved until a polynomial-time learning algorithm is discovered or until a representation-independent hardness result is proved." 
 
 Basing the hardness of representation-independent learning on $\NP$-hardness has proved much more difficult. Existing lower bounds in the representation-independent setting all rely on cryptographic and average-case assumptions. The literature on such lower bounds is equally large; we refer the reader to the pioneering papers by Kearns and Valiant~\cite{KV94} (on cryptographic assumptions) and Daniely, Linial, and Shalev-Schwartz~\cite{DLSS14} (on more general average-case assumptions).


\subsubsection{\Cref{thm:NP intro} $\Rightarrow$ A converse to~\Cref{fact:if RP=NP}}

\Cref{thm:NP intro} gives the first representation-independent lower bounds that are based on $\NP$-hardness. Our concept class~$\mathcal{C}$ is $\NP$-enumerable, and we not only rule out polynomial-time learners that use $O(\VCdim(\mathcal{C}))$ many samples, but even those that use $\Phi(\VCdim(\mathcal{C}))$ many samples for {\sl any} function~$\Phi$. \Cref{thm:NP intro} therefore yields a strong converse to \Cref{fact:if RP=NP}, and together they show an equivalence between $\RP$ vs.~$\NP$ question and the hardness of learning: 

\medskip
 
\begin{tcolorbox}[colback = white,arc=1mm, boxrule=0.25mm]
\vspace{5pt}
\begin{corollary}[An equivalence between hardness of $\NP$ and hardness of learning]
    \label{cor:dichotomy} \ 
\begin{itemize}
    \item[$\circ$] {\bf \cite{PV88}:} If $\RP = \NP$, every $\NP$-enumerable class $\mathcal{C}$ can be learned in polynomial time with $O(\VCdim(\mathcal{C}))$ samples.
    \item[$\circ$] {\bf \Cref{thm:NP intro}:} If $\RP \ne \NP$, there is an $\NP$-enumerable class $\mathcal{C}$ that cannot be learned in polynomial time with $\Phi(\VCdim(\mathcal{C}))$ samples for  any function $\Phi$.  
\end{itemize}
\end{corollary}
\vspace{1pt}
\vspace{-5pt}
\end{tcolorbox}\medskip 


\Cref{cor:dichotomy} therefore shows that the $\RP$ vs.~$\NP$ question can be recast as an independently interesting question about the power of time- and sample-efficient learners.

\paragraph{The~\cite{ABX08} barrier.}
\label{subsec:includes-ABX-intro}
As mentioned, Applebaum, Barak, and Xiao~\cite{ABX08}, building on~\cite{FF93,BT06,AGGM06}, established a formal barrier to proving a statement of the following form: ``If $\RP\ne \NP$, then $\mathcal{C}$ cannot be improperly learned in polynomial time". They showed that the polynomial hierarchy collapses if such a statement is proved using either a Karp or a non-adaptive Turing reduction. We refer the reader to~\cite{Xia09} for an excellent exposition of~\cite{ABX08}.

Our proof of~\Cref{thm:NP intro} uses a non-adaptive Turing reduction, but we are able to sidestep their barrier with a reduction that is specifically tailored to learners that use a certain number of samples. The precise statement of~\cite{ABX08}'s barrier and the reasons why our reduction sidesteps it are somewhat technical; we defer the details to~\Cref{sec:ABX}.

\paragraph{Learnability beyond polynomial-size circuits.} The PAC model requires learners to output an efficient hypothesis, i.e.~a polynomial-size circuit. Without this restriction, every concept class is time- and sample-efficiently learnable by offloading all the computation to the hypothesis (see e.g.~Exercise 2.4 in \cite{KVBook}). Given this restriction, Schapire~\cite{Sch90} used boosting to show the impossibility of learning any strict {\emph{super}}class of polynomial-size circuits. 
 For example, if $\NP/\mathsf{poly} \ne \P /\mathsf{poly}$ then  $\NP/\mathsf{poly}$ is not learnable.

The focus of much of learning theory, including our work and the~\cite{ABX08} barrier, has therefore been on learning  {\emph{sub}}classes of polynomial-size circuits. In this setting, lower bounds are more difficult to establish. While Schapire's result is representational in nature, concerning the impossibility of representing a superclass of polynomial-size circuits as polynomial-size circuits, our result can no longer rely on representational impossibility. We are instead concerned with the computational difficulty of finding a good hypothesis even with the promise that one exists.




\subsection{Extensions to other settings}

Our proof extends to give analogous tradeoffs in other well-studied settings. 

\paragraph{Uniform-distribution learning.} The focus of our exposition is on Valiant's original model of distribution-free PAC learning, where  learners are expected to succeed with respect to all distributions. However, our construction and proof can be modified so that they also give tradeoffs for {\sl uniform-distribution} PAC learning.  In fact, in the modified construction the upper bounds are witnessed by learners that succeed with respect to all distributions, and the lower bound holds even if the learner only has to succeed with respect to the uniform distribution. See~\Cref{thm:uniform}. 


There is a sense in which uniform-distribution learning is ``more average-case" in nature than distribution-free learning. This makes lower bounds, especially those based on worst-case assumptions, more difficult to prove. In the distribution-free setting hardness reductions can (and often) produce  distributions that place their mass on a few carefully chosen inputs, but  the uniform-distribution setting does not admit this degree of freedom. 

\paragraph{Online learning.} Two of the most standard theoretical models for supervised learning are the PAC model (batch learning) and Littlestone's mistake bound model (online learning)~\cite{Lit88}.  We also give analogous tradeoffs for the latter, where Littlestone dimension takes the place of VC dimension and mistake bounds take the place of sample complexity.  See~\Cref{thm:online}.


 There has long been a focus on time-efficient online algorithms with optimal mistake bounds. Indeed, there are two parts to Littlestone's original paper: the first proves that Littlestone dimension characterizes the optimal mistake bound  {\sl if runtime is not a concern} (exactly analogously with VC dimension and sample complexity), and the second gives a time-efficient algorithm, {\sc Winnow}, that achieves the optimal mistake bound for simple concept classes. Our work gives the first example of a concept class that, assuming $\RP \ne \NP$, does not admit any such algorithm.

\section{Technical Overview}



Our approach is simple. For each language $L \in \NP$, we define a learning problem that inherits its time vs.~sample complexity tradeoffs from the {\sl time vs.~nondeterminism} tradeoffs for $L$. Crucial to this connection is a view of nondeterminism as a  quantitative resource~\cite{KF77,GLM96}. While conjectures such as $\RP \ne \NP$ and  randomized ETH concern the distinction between algorithms with no access to nondeterminism versus those with unlimited access to nondeterminism, we will be concerned with the finer-grained distinction between algorithms with access to a few bits of nondeterminism versus those with access to more. 

\paragraph{Time vs.~nondeterminism tradeoffs for $\mathrm{SAT}$.} For concreteness, consider $L = \mathrm{SAT}$. 
Instances of $\SAT$ over $n$ variables can be decided in polynomial time with $n$ bits of nondeterminism: these $n$ bits specify a satisfying assignment. The conjecture that $\RP\ne \NP$ posits that the time complexity of $\SAT$ is superpolynomial for (randomized) algorithms with no access nondeterminism, and similarly,  randomized ETH posits that it is exponential for such algorithms. 

What about algorithms with access to $b(n)$ bits of nondeterminism for $0 < b(n) < n$? Surely giving algorithms access to more bits of nondeterminism can only lead to faster runtimes, but what is the rate of speedups? In other words, we are interested in how the quantity:  
\[ \SAT(b(n)) \coloneqq \text{Time complexity of $\SAT$ given $b(n)$ bits of nondeterminism} \]
behaves as $b(n)$ ranges from $0$ to $n$.
We call this the {\sl time vs.~nondeterminism tradeoff curve} for $\SAT$. 

\subsection{A learning problem based on $\SAT$} 

We now define a learning problem that inherits its time vs.~sample complexity tradeoffs from the curve $\SAT(b(n))$. For an instance $\phi$ of $\SAT$ over $n$ variables, let $\mathrm{Sol}(\phi) \in \zo^n$ denote the lexicographically first satisfying assignment of $\phi$ if $\phi \in \SAT$ and the all-zeroes string otherwise. We define a corresponding boolean function $f_\phi$: 
\[ f_\phi(x) = 
\begin{cases}
\mathrm{Enc}(\mathrm{Sol}(\phi))_i & \text{if $x = (\phi, i)$} \\
0 & \text{otherwise,} 
\end{cases}
\]
where $\mathrm{Enc}(\cdot)$ is the encoding function of an efficient error correcting code with constant rate and distance. Let us refer to inputs $x$ whose prefix match $\phi$ (i.e.~$x = (\phi,i)$ for some $i$) as {\sl useful} and all other inputs as {\sl useless}. Therefore~$f_\phi$ reveals a coordinate of $\mathrm{Enc}(\mathrm{Sol}(\phi))$ on useful inputs and always outputs $0$ on useless ones. 

\paragraph{$\mathcal{C}_{\SAT}$ and its basic properties.} Now consider the concept class $\mathcal{C}_{\SAT}$ comprising of $f_\phi$ for all instances $\phi$ of $\SAT$. We begin with two  claims about $\mathcal{C}_{\SAT}$, corresponding to~\Cref{thm:NP intro:VCdim} and~\ref{thm:NP intro:many samples}. The first is that  $\VCdim(\mathcal{C}_{\mathrm{SAT}}) = 1$. The intuition is that the learner knows $\phi$ after receiving a {\sl single} useful example, and if time-efficiency is not a concern it can then determine  $\mathrm{Sol}(\phi)$ by brute force, thereby learning $f_\phi$ exactly. The second claim is that $\mathcal{C}_{\SAT}$ is time-efficiently learnable with $O(n)$ samples. This is because with that many samples, the learner will have seen most coordinates of $\mathrm{Enc}(\mathrm{Sol}(\phi))$ and can therefore efficiently produce an accurate hypothesis. More work is needed to make these proofs formal (the last sentence for example assumes that the distribution over examples places significant weight on useful ones), but this description captures the essential ideas. 

$\mathcal{C}_{\SAT}$ is therefore learnable  in exponential time with $O(1)$ samples and in polynomial time with $O(n)$ samples.  As in the case of $\SAT$ above, we are now interested in how the quantity 
\[ \mathcal{C}_{\SAT}(s(n)) \coloneqq \text{Time complexity of learning $\mathcal{C}_{\SAT}$ with $s(n)$ samples.} \]
behaves as $s(n)$ ranges from $0$ to $O(n)$, i.e.~the {\sl time vs.~sample complexity tradeoff curve} for $\mathcal{C}_{\SAT}$. 

\subsection{$\SAT$'s curve lowerbounds $\mathcal{C}_{\SAT}$'s curve} The discussion so far already hints at a connection between the complexity of deciding $\SAT$ and that of learning $\mathcal{C}_{\SAT}$. On one hand, $\SAT$ can be decided in exponential time with no bits of nondeterminism and $\mathcal{C}_{\SAT}$ can be learned in exponential time with $O(1)$ samples. On the other hand,  $\SAT$ can be decided in polynomial time with $n$ bits of nondeterminism and $\mathcal{C}_{\SAT}$ can be learned in polynomial time with $O(n)$ samples.

Indeed, we show that the connection extends to all intermediate settings of parameters: \medskip



\begin{lemma}[Informal; see~\Cref{lem:one curve lower bounds the other} for the formal version]
\label{lem:one curve lower bounds the other intro} For all growth functions $s(n)$,
    \[ \mathcal{C}_{\SAT}(s(n)) \succsim \frac{\SAT(s(n))}{n\log n} - \poly(n).\]
\end{lemma}

In words,~\Cref{lem:one curve lower bounds the other intro} says that the task of learning $\mathcal{C}_{\SAT}$ inherits its time vs.~sample complexity tradeoffs from the time vs.~nondeterminism tradeoffs of $\SAT$. Pictorially, we visualize this as the tradeoff curve for $\SAT$ (essentially) lowerbounding the tradeoff curve $\mathcal{C}_{\SAT}$. See~\Cref{fig:lower-bound-curve}. \Cref{lem:one curve lower bounds the other intro} implies, for example, that if the time complexity of deciding $\SAT$ with $s(n)$ bits of nondeterminism is superpolynomial, then so is the time complexity of learning $\mathcal{C}_{\SAT}$ with $s(n)$ many samples.

\begin{figure}[h!]
    \hspace{0.6em}
    \begin{tikzpicture}
        \def\l{-0.4}
        \def\lp{-0.3}
        
        \def\hOffset{0.07}  
        \def\vOffset{0.1}  
        
        \begin{axis}[
            scale only axis,
            width=0.6\linewidth,
            xmajorticks=false,
            ytick={0.1,0.9},
            yticklabels={$\poly(n)$,$2^{O(n)}$},
            xmin=\l, xmax=1.4,
            ymin=0.05, ymax=1.1,  
            axis lines=left,
            clip=false,
            legend pos=south west,
            legend style={nodes={scale=0.9, transform shape}},
            xlabel = {Number of nondeterministic bits / samples},
            ylabel = {Runtime},
            y label style = {rotate=-90},
            y label style = {at={(-0.19,0.5)}},
            x label style = {at={(0.45,-0.05)}},
        ]
        \addplot[color=purple, smooth, thick, domain=\lp:0.651] {0.899};
        \addplot[color=purple, smooth, thick, domain=0.9:1.4, forget plot] {0.104};
        \addplot[color=purple, smooth, thick, domain=0.65:0.9, forget plot] {-0.0045*atan(300*(x-0.8))+0.5};

        \addplot[color=blue, smooth, thick, domain=\lp:0.651+\hOffset] {0.899 + \vOffset};  
        \addplot[color=blue, smooth, thick, domain=0.9+\hOffset:1.4] {0.104 + \vOffset};   
        \addplot[color=blue, smooth, thick, domain=0.65+\hOffset:0.9+\hOffset] {-0.0045*atan(300*((x-\hOffset)-0.8))+0.5 + \vOffset}; 
        
        \draw[-Stealth] (0.25,0.899) to (0.25,0.95);

        
        \draw[] (0.22,0.85) node [fill=white,text width=5.5cm, centered] {{\footnotesize{SAT's time vs.~nondeterminism curve}}};
        \draw[] (0.35,1.05) node [fill=white,text width=5.5cm, centered] {\footnotesize{$\mcC_{\text{SAT}}$'s time vs.~samples curve}};
    \end{axis}
    \end{tikzpicture}
    \caption{An illustration of~\Cref{lem:one curve lower bounds the other intro}}
    \label{fig:lower-bound-curve}
\end{figure}
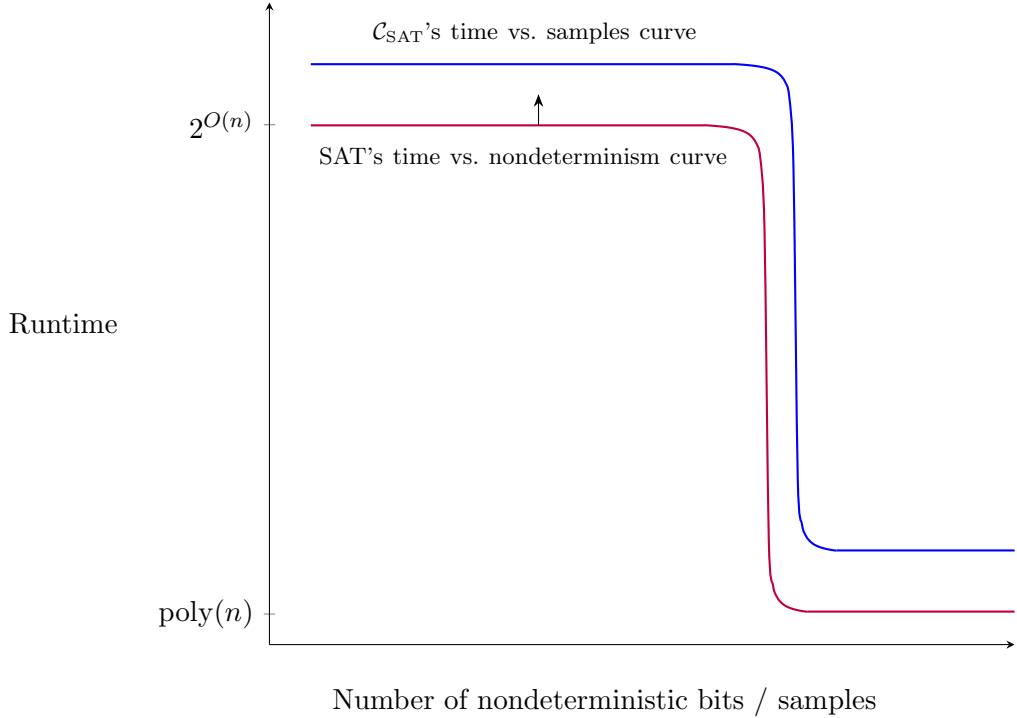


The proof of~\Cref{lem:one curve lower bounds the other intro} proceeds via the contrapositive: we show how a time- and sample-efficient algorithm for learning $\mathcal{C}_{\SAT}$ yields a time- and nondeterminism-efficient algorithm for deciding $\SAT$. To do so, we give a nondeterministic reduction from deciding $\SAT$ to learning $\mathcal{C}_{\SAT}$, where the number of bits of nondeterminism the reduction uses corresponds to the sample complexity of the learning algorithm.


\Cref{thm:NP intro:few samples} follows   straightforwardly from~\Cref{lem:one curve lower bounds the other intro} and the fact that an algorithm using $b$ bits of nondeterminism can be simulated by one using no nondeterminism with a runtime overhead of~$2^b$.  However, our description  omits a slight technical complication. Owing to the inherent randomness of the PAC model (specifically, the distribution over examples) 
and the internal randomness of PAC learners,~\Cref{lem:one curve lower bounds the other intro} actually lowerbounds $\mathcal{C}_{\SAT}$'s curve in terms of the {\sl Arthur-Merlin}  time complexity of deciding~$\SAT$ with a certain number of nondeterministic bits.  ($\AM$ is a randomized variant of $\NP$~\cite{Bab85}.) See~\Cref{lem:one curve lower bounds the other}, the formal version of~\Cref{lem:one curve lower bounds the other intro}. This is ultimately why~\Cref{thm:NP intro:few samples} is based on the hardness of simulating $\NTIME$ with $\RTIME$ instead of $\mathsf{DTIME}$. 

\begin{remark}[Obtaining the full range of parameters in~\Cref{thm:NP intro}] 
\label{rem:full range} 
This construction cannot give a larger gap than $O(1)$ vs.~$\Theta(n)$ samples since $\SAT$ {\sl can} be efficiently decided with $n$ bits of nondeterminism. To get an arbitrarily large gap of $O(1)$ vs.~$\Theta(p(n))$, we therefore look beyond $\SAT$ to languages in $\NP$ that use $p(n)$ bits of nondeterminism. The ideas in this section extend  to show, for any language $L \in \NP$, the existence of a concept class $\mathcal{C}_L$ that inherits its time vs.~sample complexity tradeoffs from the time vs.~nondeterminism tradeoffs for $L$. The sought-for tradeoffs for $\mathcal{C}_L$ then hold under corresponding hardness assumptions about~$L$. 
\end{remark}

\subsection{Details of the~\cite{ABX08} barrier and why we evade it}
\label{sec:ABX}

\cite{ABX08} considered {\sl nonadaptive Turing reductions} from an $\NP$-hard language $L$ to the learning of a concept class $\mathcal{C}$. These are  reductions that map an instance $z$ to multiple, say~$k$, learning instances such that the hypotheses for these instances allow one to efficiently determine if $z \in L$. In more detail, they considered reductions that produce generators $G_1,\ldots,G_k$ of distributions over $\zo^n\times \zo$. On a random seed $\br$, the $i$-th generator $G_i(\br)$ outputs a labeled example $(\bx,\by)$. We associate $G_i$ with the function $f_i$ whose labeled examples it generates, and the corresponding learning task is to output an accurate hypothesis $h_i$ for $f_i$ if $f_i \in \mathcal{C}.$  The reduction is allowed to produce $G_i$'s whose corresponding $f_i$ is not in $\mathcal{C}$, but in this case the learner is off the hook and there are no requirements on its performance. \cite{ABX08} further required the reduction to decide if $z \in L$ by accessing the $h_i$'s only as blackboxes. They proved that such a nonadaptive blackbox Turing reduction can be used to collapse the polynomial hierarchy to the second level. Since the polynomial hierarchy is believed to be infinite, this suggests that such a reduction does not exist.\footnote{\cite{ABX08} also established barriers for more general types of reductions (constant-round Turing reductions) and showed that their barriers even hold in the setting of agnostic learning (where $f_i$ only has to be close to $\mathcal{C}$), but it is this result that is most relevant to us.} 

Our proof of~\Cref{thm:NP intro} can be cast as a nonadaptive blackbox Turing reduction and  therefore falls within~\cite{ABX08}'s framework.  However, we evade their barrier because our reduction only works for learners that access a certain number of samples from the $G_i$'s. Our reduction actually produces $G_i$'s for which {\sl none} of the corresponding $f_i$'s belong to $\mathcal{C}$. Crucially, though, since we are dealing with learners with a bound $m$ on their sample complexity, it suffices for us to show that we are likely to ``get lucky" in the sense of having one of the $G_i$'s produce a batch of $m$ labeled examples that happen to coincide with some function in $\mathcal{C}$. This guarantee no longer holds, and consequently our overall reduction breaks, for learners that are allowed $\gg m$ samples---necessarily so, as otherwise the~\cite{ABX08} barrier would apply.

\section{Discussion and Future Work}
\label{sec:future work}

Understanding the relative strengths of cryptographic, average-case, and worst-case hardness assumptions, the subject of Impagliazzo's Five Worlds survey~\cite{Imp95fiveworlds}, is a central goal of complexity theory. Pinning down the minimal assumptions necessary for computational lower bounds in learning has been challenging because there is {\sl both} a worst-case and an average-case aspect to the PAC model:  the learner has to succeed with respect to a worst-case concept in the concept class, and yet its hypothesis only has to achieve high accuracy with respect to a distribution over examples.  

It has long been noted~\cite{IL90,BFKL93} that cryptographic assumptions, which are fundamentally average-case in nature, yield learning lower bounds that are ``stronger than necessary": not only do they show that a {\sl worst-case} concept is hard to learn, they even give a {\sl distribution} over concepts under which learning is average-case hard. Quoting~\cite{BFKL93}, ``This state of affairs naturally leads one to wonder
if intractability results for learning in fact {\sl require}
cryptographic assumptions." On the other hand, as formalized in~\cite{ABX08}, the distribution over examples makes it challenging to base learning lower bounds on worst-case assumptions.  

Our work makes progress towards resolving this tension by showing that it {\sl is} possible to obtain computational lower bounds for learning from $\NP$-hardness---if sample complexity is also taken into account. Sample-efficiency is as important a desideratum as time-efficiency, and so it is of interest to understand the limitations of time- and sample-efficient algorithms.
One avenue for future work is to see if our techniques could apply to concept classes with long-conjectured computational-statistical tradeoffs (e.g.~decision lists~\cite{Blu90}). Another is to understand if the same is true for the many other settings, beyond PAC learning, in which such tradeoffs are being studied.










\section{Preliminaries}

\paragraph{Notation and naming conventions.} We write $[n]$ to denote the set $\{1,2,\ldots,n\}$. We use lowercase letters to denote bitstrings e.g. $x,y\in\zo^n$ and subscripts to denote bit indices: $x_i$ for $i\in [n]$ is the $i$th index of $x$. The length of a bitstring is $|x|$. We use \textbf{boldface} letters, e.g.~$\bx,\by$, to denote random variables.  For functions $f,g:\zo^n\to\zo$ and distributions $\D$, we write $\error_{\D}(f,g)$ to denote $\Prx_{\bx\sim\D}[f(\bx)\neq g(\bx)]$. When $\D$ is the uniform distribution, we drop the subscript and write $\error(f,g)$. The \textit{sparsity} of a function $f:\zo^n\to\zo$ refers to the number of $1$-inputs. A function is $k$-sparse if $|f^{-1}(1)|\le k$.

\subsection{Learning theory}

\begin{definition}[PAC learning \cite{Val84}]
    \label{def:PAC-learning}
    For any concept class $\mcC$, we say an algorithm $A$ \emph{learns $\mcC$ to error $\eps$ over distribution $\mathcal{D}$ using $m$ samples} if the following holds: For any $f \in \mcC$, given $m$ independent samples of the form $(\bx, f(\bx))$ where $\bx \sim \mathcal{D}$, $A$ returns a hypothesis $h$, that with probability at least $2/3$, satisfies $\error_{\D}(f,h)\le \eps.$
    
    We say that $A$ learns $\mcC$ to error $\eps$ using $m$ samples if \emph{for every} distribution $\D$, $A$ learns $\mcC$ to error $\eps$ over the distribution $\mathcal{D}$ using $m$ samples.
\end{definition}

\begin{definition}[VC dimension \cite{VC71}]
    Let $\mcC$ be a concept class consisting of Boolean functions. A set $S\sse \zo^n$ is \emph{shattered} if for every labeling $\ell:S\to\zo$, there is an $f\in\mcC$ such that $f(x)=\ell(x)$ for every $x\in S$. The \emph{VC dimension} of $\mcC$, $\VCdim(\mcC)$, is the size of the largest set which is shattered. 
\end{definition}

A basic fact about VC dimension is that it can be upperbounded in terms of the size of the concept class: 
\begin{fact}
\label{fact:vc-upper}
    For a finite concept class $\mcC$, we have $\VCdim(\mcC)\le\log |\mcC|$.
\end{fact}

A \emph{empirical risk minimizer} for a concept class $\mcC$ is an algorithm that draws a random sample and outputs a hypothesis $h\in\mcC$ with lowest error over the sample. In our setting, the examples are labeled by a concept $f\in\mcC$, and so an empirical risk minimizer always outputs some hypothesis consistent with the examples. With enough samples, any empirical risk minimizer PAC learns $\mcC$:

\begin{fact}[Empirical risk minimizer is a PAC learner \cite{BEHW89}]
\label{fact:consistent}
    Let $\mcC$ be a finite concept class. An empirical risk minimizer learns $\mcC$ to error $\eps$ for any $\eps>0$ with sample complexity $\Theta(\log|\mcC|/\eps)$.
\end{fact}

\subsection{Complexity theory} 
\subsubsection{Randomized time}
\begin{definition}[Randomized time with one sided error]
    A language $L$ is in $\RTIME(t(n))$ if there is a randomized algorithm $A$ running in worst-case time $O(t(n))$ with the following properties. For all $x\in\zo^*$, $A$ takes as input a random string $\br$ and $x$ and satisfies:
    \begin{enumerate}
        \item if $x\in L$, then $\Prx_{\br}[A(\br,x)=1]\ge 2/3$;
        \item if $x\not\in L$, then $\Prx_{\br}[A(\br,x)=0]=1$.
    \end{enumerate}
\end{definition}
This definition is a refinement of the standard complexity class $\mathsf{RP}$ which is the class of languages in $\RTIME(t(n))$ for some $t(n)=\poly(n)$.

The randomized exponential time hypothesis states that randomized algorithms for SAT require exponential time.
\begin{hypothesis}[Randomized exponential time hypothesis (ETH) \cite{CIKP03,DHHTMTW14}]
There exists a constant $\delta>0$ such that $3$-\textnormal{SAT} on $n$ variables cannot be solved in randomized time $O(2^{\delta n})$.
\end{hypothesis}

\subsubsection{Complexity classes involving nondeterminism}
\begin{definition}[Nondeterministic time and nondeterministic time with bounded certificate size]
    The class of languages decided in nondeterministic time $O(t(n))$ is $\NTIME(t(n))$. The class of languages which can be decided by a verifier running in time $O(t(n))$ with certificates of size $O(p(n))$ is $\NTIME(t(n),p(n))$.
\end{definition}
The definition of $\NTIME(t(n),p(n))$ is a refinement of the standard definition $\NTIME(t(n))$ which is itself a refinement of $\NP$. Since $p(n)\le t(n)$ by definition, we have $\NTIME(t(n))=\NTIME(t(n),t(n))$. Furthermore, $\NP$ is the union of all languages in $\NTIME(t(n))$ for some $t(n)=\poly(n)$. 

We will also use a similar refinement of the complexity class $\AM$.

\begin{definition}[$\AM$~protocols with bounded proof size]
    A language $L$ is in $\AMTIME(t(n),p(n))$ if it is decided by an Arthur-Merlin $(\AM)$ protocol where Arthur runs in time $O(t(n))$ and uses proofs of length $O(p(n))$ from Merlin. Formally, an $\AM$~protocol is a Turing machine $M$ which takes an input $z$, $r$ random bits, and a proof $w$, and decides whether $z\in L$. The protocol $M$ must satisfy \textnormal{completeness} and \textnormal{soundness}.
    \begin{enumerate}
        \item \textbf{Completeness}: If $z\in L$, then
        $$
        \Prx_{\br}[\textnormal{there exists a }w\textnormal{ such that }M(\br,w,z)=1]\ge 2/3.
        $$
        \item \textbf{Soundness}: If $z\not\in L$, then 
        $$
        \Prx_{\br}[\textnormal{there exists a }w\textnormal{ such that }M(\br,w,z)=1]\le 1/3.
        $$
    \end{enumerate}
The protocol has \textnormal{perfect soundness} if the corresponding probability is $0$ for every $z\not\in L$. Furthermore, if the runtime of $M$ is $O(t(n))$ and $|w|\le O(p(n))$, then $L\in \AMTIME(t(n),p(n))$.
\end{definition}

Informally, an $\AM$ protocol operates by first having Arthur, the verifier, send a random string $\br$ to Merlin who then provides a proof $w$ that $z\in L$. If $z$ is in fact in $L$, then, with high probability over $\br$, there must exist some proof that convinces Arthur, and otherwise all proofs must fail. 

In this paper, all of the $\AM$ protocols will have perfect soundness. 

An important difference between $\AMTIME(t(n),p(n))$ and the standard class $\AM$ is that for $L\in \AM$ one can assume the $\AM$ protocol for $L$ has perfect completeness. This is because one can convert a protocol with imperfect completeness into one with perfect completeness by increasing the size of the proof used by Merlin. This transformation does not apply to the class $\AMTIME(t(n),p(n))$.

\subsubsection{Error-correcting codes}

\begin{definition}[Efficient binary error-correcting codes] 
\label{def:codes}
An \emph{efficient binary error-correcting code} is a pair of polynomial-time algorithms, the encoder $(\Enc)$ and the decoder $(\Dec)$, with the following properties. There are constants $c>1$ and $\eps^\star >0$ such that $\Enc:\zo^n\to \zo^{c n}$ and $\Dec:\zo^{c n}\to \zo^{n}$, and given $y$ such that $y=\Enc(x)$, the decoder $\Dec$ can decode messages which have up to $\eps^\star c n$ errors, i.e.~$\Dec(\Tilde{y})=x$ for any $\Tilde{y}$ which differs from $y$ on at most $\eps^\star c n$ coordinates. The \emph{rate} of the code is $1/c$.
\end{definition}

Binary error-correcting codes of this form exist. For example, one could use expander codes~\cite{SS96}. Using the explicit construction of expanders in \cite{CRVW02}, for every $c>1$ there is a sufficiently small $\eps^\star>0$, such that efficient binary error-correcting codes with parameters $c$ and $\eps^\star$ exist. 


\subsubsection{Oracles and enumerability}

\begin{definition}[Oracle Turing machines]
    A \emph{oracle Turing machine} $M$ is a standard Turing machine augmented with a \emph{oracle tape}. Given access to an \emph{oracle} $\mcO: \zo^{\star} \to \zo$, at unit cost, $M^{\mcO}$ may determine the value of $\mcO(x)$ where $x$ is the quantity currently written onto its oracle tape. \emph{Nondeterministic oracle Turing machines} and \emph{randomized oracle Turing machines} are analogous extensions of standard nondeterministic and randomized Turing machines.
\end{definition}

\begin{definition}[Enumerability of concept classes]
    \label{def:enumerable-general}
    For any oracle $\mcO$, we say a concept class $\mcC$ is $\mcO$-enumerable if there exists an efficient oracle Turing machine $M$ and $s(n) = \poly(n)$ with the following properties for all $n \in \N$:
    \begin{enumerate}
        \item For every seed $s \in \zo^{s(n)}$, the output of $M^{\mcO}(S)$ is a circuit computing some function in $\mcC$.
        \item Every function in $\mcC$ is the output of $M^{\mcO}(s)$ for some seed $s \in \zo^{s(n)}$.
    \end{enumerate}
\end{definition}

\begin{definition}[$\NP$- and $\NTIME(t(n),p(n))$-enumerability]
    \label{def:enumerable-specific}
    We say a concept class is $\NP$-enumerable if it is $\mcO$-enumerable for an oracle $\mcO$ computing a language in $\NP$. $\NTIME(t(n),p(n))$-enumerability is the analogous definition for the class $\NTIME(t(n),p(n))$.
\end{definition}

\section{The hard concept class and its basic properties}

\subsection{A concept class $\mathcal{C}_{L,V}$ for every language $L$ and verifier $V$}

First, we define the concept class that will be essential to the proof. For every verifier $V$ which verifies a language $L\in \NTIME(t(n),p(n))$, we obtain a distinct concept class $\mathcal{C}_{L,V}$.
\begin{definition}[Definition of $\mathcal{C}_{L,V}$]
    \label{def:concept-class}
    Let $L\in \NTIME(t(n),p(n))$ for growth functions $t(n),p(n)\ge n$ and let $V$ be a time-$O(t(n))$ verifier for $L$ using size-$O(p(n))$ certificates. Let $(\Enc,\Dec)$ be the error-correcting code from \Cref{def:codes} with rate $1/c$ for constant $c>1$. For each $z\in \zo^n$, we define a concept $\Cert_z:\zo^{n+\log (c p(n))}\to \zo$ as follows. If $z\not\in L$, then $\Cert_z$ is the constant $0$ function. Otherwise, let $w^\star\in \{0,1\}^{p(n)}$ be the lexicographically first certificate such that $V(z,w^\star)=1$ and define $\Cert_z(x)$ as 
    $$
    \Cert_z(x)=\begin{cases}
        \Enc(w^\star)_i & \text{if }x=(z,i)\\
        0 & \text{otherwise}
    \end{cases}
    $$
    where $i\in \zo^{\log (cp(n))}$ is interpreted as an integer index in the range $\{1,\ldots,cp(n)\}$. We call examples $x$ where $x=(z,i)$ for some $i$ as \textnormal{useful} examples and all other examples \textnormal{useless}.  Finally, the concept class $\mathcal{C}_{L,V}$ is defined to be the set of all $\Cert_z$ for $z\in\zo^n$. 
\end{definition}

\subsection{Formal statement of \Cref{thm:NP intro}}

Using the concept class in \Cref{def:concept-class}, we prove \Cref{thm:NP intro}. The formal version of the theorem is as follows.
\begin{theorem}[Formal statement of \Cref{thm:NP intro}]
\label{thm:NP formal}
    For every time constructible growth function $p(n)\ge n$,  there is a concept class $\mcC$ such that the following holds.
    \begin{enumerate}
        \item We have $\VCdim(\mcC)=1$ and $\mcC$ is $\NTIME(p(n)\log p(n),p(n))$-enumerable. There is an algorithm that, for every $\eps>0$, learns $\mcC$ to error $\eps$ using $O(1/\eps)$ samples in time $2^{O(p(n))}$.\label{thm:NP formal:VCdim} 
        \item Let $\eps^\star>0$ be the constant from \Cref{def:codes}. If there is an $m(n)$-sample, time-$t(n)$ algorithm that learns $\mcC$ w.h.p.~to error $\eps^\star$, then $\NTIME(p(n))\sse \RTIME\left(2^{O(m(n))}t(n)\log t(n)\poly(p(n))\right)$.\label{thm:NP formal:few samples}
        \item There is an algorithm that, for every $\eps>0$, learns $\mcC$ to error $\eps$ whose sample complexity and runtime are both $O(p(n)/\eps)$.\label{thm:NP formal:many samples}
    \end{enumerate}
    Furthermore, each concept in $\mcC$ is computed by a circuit of size $O(p(n))$, and in fact, even a decision tree of size $O(p(n))$.
\end{theorem}

Note that by choosing $m(n)=O(\log t(n))$ in \Cref{thm:NP formal:few samples}, the collapse becomes $\NTIME(p(n))\sse \RTIME\left(\poly(t(n),p(n))\right)$. The statement of \Cref{thm:NP intro:few samples} then follows via the contrapositive. 

\subsection{Basic properties of the concept class: Proofs of \Cref{thm:NP formal:VCdim} and \Cref{thm:NP formal:many samples}}
\label{subsec:basic-properties-of-C}

\subsubsection{\Cref{thm:NP formal:VCdim}: The VC dimension of the concept class and its enumerability}

We show the VC dimension of $\mathcal{C}_{L,V}$ is $1$ and that it can be learned to error $\eps$ with $O(1/\eps)$ samples. 
\begin{claim}
\label{claim:vc-dimension}
    Let $L\in \NTIME(t(n),p(n))$ for growth functions $t(n),p(n)\ge n$ and let $V$ be a time-$O(t(n))$ verifier for $L$ using size-$O(p(n))$ certificates. Then $\VCdim(\mathcal{C}_{L,V}) = 1$. Furthermore, there is an algorithm that, for every $\eps>0$, learns $\mathcal{C}_{L,V}$ to error $\eps$ using $O(1/\eps)$ samples in time $O(t(n))\cdot 2^{O(p(n))}$.
\end{claim}

\begin{proof}
    Let $z\in L$ and $w^\star$ be the lexicographically first certificate such that $V(z,w^\star)=1$. Let $i\in \zo^{\log (cp(n))}$ be such that $\Enc(w^\star)_i=1$. Then clearly the set $\{(z,i)\}$ is shattered. Now, we claim that no set of 2 points is shattered. Consider the points $(z^{(1)},i)$ and $(z^{(2)}, j)$. If $z^{(1)}\neq z^{(2)}$, then no concept labels both points as $1$. Otherwise, the only labelings realized are $(0,0)$ and $(\Enc(w^\star)_i,\Enc(w^\star)_j)$ where $w^\star$ is the lexicographically first certificate for $z^{(1)}=z^{(2)}$. Therefore, these points are not shattered.
    
    For the second part of the claim, every $1$-input to a concept $\Cert_z\in\mathcal{C}_{L,V}$ has the form $(z,i)$. If the learning algorithm receives a $1$-input, it can use $z$ to find the lexicographically first certificate $w^\star$ such that $V(z,w^\star)=1$ and then output the exact representation of $\Cert_z$, all in time $O(t(n))\cdot 2^{O(p(n))}$. Using this observation, it is straightforward to show that $O(1/\eps)$ random samples suffice to learn $\Cert_z$ over any distribution. Specifically, for any distribution $\D$ over $\zo^n$, if $\Pr_{\bx\sim\D}[\Cert_z(\bx)=1]\ge \eps$, then for any random sample $\bS$ of size $m$, we have
    $$
    \Prx_{\bS}[\Cert_z(\bx^{(i)})=0\text{ for all examples }\bx^{(i)}\text{ in }\bS ]\le (1-\eps)^m\le\exp(-\eps m).
    $$
    By choosing $m=O(1/\eps)$, we have that with probability at least 0.99, some example in $\bS$ is labeled $1$ and our learning algorithm outputs an exact representation of $\Cert_z$. If no example is labeled $1$, the learner outputs the constant $0$ function. When $\Pr_{\bx\sim\D}[\Cert_z(\bx)=1]< \eps$, the constant $0$ function approximates $\Cert_z$ to error $\eps$. So in either case, our learning algorithm outputs a hypothesis with error $\eps$ with probability at least $0.99$.
\end{proof}

The rest of the section is devoted to proving the following claim which establishes the enumerability of $\mathcal{C}_{L,V}$.

\begin{claim}[Enumerability of $\mathcal{C}_{L,V}$]
    \label{claim:enumerable}
    Let $L$ be a language in $\NTIME(t(n),p(n))$ for growth functions $t(n),p(n)\ge n$ and let $V$ be a time-$O(t(n))$ verifier for $L$ using size-$O(p(n))$ certificates. There is an algorithm which takes as input $z\in\zo^n$, makes $p(n)$ calls to an $\NTIME(t(n),p(n))$ oracle, and outputs a decision tree representation of $\Cert_z\in\mathcal{C}_{L,V}$ in  $\poly(p(n))$ time.
\end{claim}

 This result implicitly shows that every concept in $\mathcal{C}_{L,V}$ can be represented as a decision tree of size $\poly(p(n))$. It also implies that the class $\mathcal{C}_{L,V}$ is enumerable with an oracle for $\NTIME(t(n),p(n))$. \Cref{thm:NP formal:VCdim} will follow from this along with the fact that $\VCdim(\mathcal{C}_{L,V}) = 1$.

Before we prove \Cref{claim:enumerable}, we show that $\Cert_z$ has a small decision tree representation which can be efficiently computed given the lexicographically first certificate for $z$. 

\begin{proposition}[$\Cert_z$ has a small decision tree representation]
\label{prop:cert_dt}
    Let $\Cert_z$ and $p(n)$ be as in \Cref{def:concept-class}. For every $z\in \zo^n$, $\Cert_z$ is computed by a decision tree $T_z$ of size $O(p(n))$. Furthermore, if $z\in L$ and $w^\star$ is the lexicographically first certificate for $z$, then there is a polynomial-time algorithm which, given $z$ and $w^\star$, outputs $T_z$.
\end{proposition}

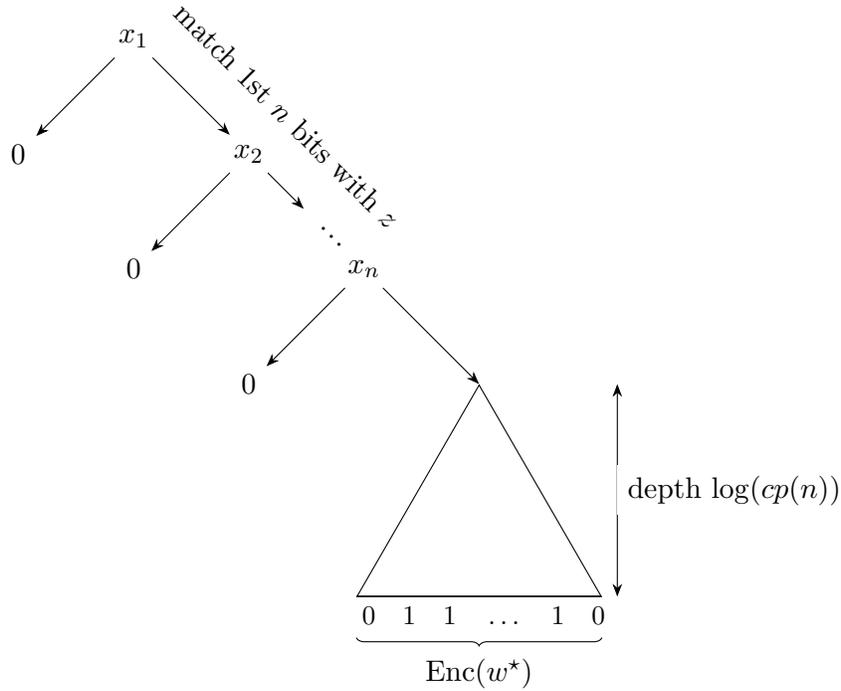
\begin{figure}[h!]
    \centering
    \resizebox{0.7\linewidth}{!}{
    \begin{tikzpicture}[tips=proper]
        \draw[] (2,2) node [fill=white] (x1) {$x_1$};
        \draw[] (3.5,0.5) node [fill=white] (x2) {$x_2$};
        \draw[] (4.5,-.5) node [fill=white,rotate=45] (x3) {$\vdots$};
        \draw[] (5,-1) node [fill=white] (xn) {$x_n$};

        \draw[] (0.5,0.5) node [fill=white] (z1) {$0$};
        \draw[] (2,-1) node [fill=white] (z2) {$0$};
        \draw[] (3.5,-2.5) node [fill=white] (zn) {$0$};

        \draw[-Stealth] (x1) to (x2);
        \draw[-Stealth] (x2) to (x3);

        \draw[-Stealth] (x1) to (z1);
        \draw[-Stealth] (x2) to (z2);
        \draw[-Stealth] (xn) to (zn);
                
        \node[isosceles triangle,
            draw,
            anchor=apex,
            isosceles triangle apex angle=60,
            rotate=90,
            minimum size=2.75cm] (T1) at (6.5,-2.5){};
        \draw[-Stealth] (xn) to (T1.apex);

        \draw[] (4,1) node [fill=white,rotate=-45] (x_label) {{match 1st $n$ bits with $z$}};

        \draw[{Stealth[scale=1]}-{Stealth[scale=1]}] ([xshift=1.8cm]T1.east) to node[right,fill=white!30,scale=1] {depth $\log (c p(n))$} ([xshift=1.8cm]T1.west);

        \draw [] ([xshift=0cm]T1.left corner) -- (T1.right corner) node [black,pos=0.5,below] { \small{$~ 0\quad 1\quad 1\quad \ldots\quad 1\quad 0$}};
        
        \draw [black,decorate,decoration={brace,mirror,raise=1pt,amplitude=3pt}] ([xshift=0cm,yshift=-0.5cm]T1.left corner) -- ([yshift=-0.5cm]T1.right corner) node [black,pos=0.5,yshift=-0.5cm] {$\Enc(w^\star)$};
        
    \end{tikzpicture}
    }
    \caption{Illustration of a decision tree for $\Cert_z$ where $z=1^n$.}
    \label{fig:dt for cert}
\end{figure}

\begin{proof}
    If $z\not\in L$, then $\Cert_z$ is a decision tree of size $1$. Otherwise, assume $z\in L$ and $w^\star$ is the lexicographically first certificate for $z$. We will describe the decision tree $T_z$. See \Cref{fig:dt for cert} for an illustration. First, $T_z$ fully queries the first $n$ bits of input and checks that they match $z$. Formally, for each $i\in [n]$, $T_z$ queries $x_i$ and outputs $0$ if $x_i\neq z_i$ and moves on to querying $x_{i+1}$ otherwise. Then, $T_z$ fully queries the next $\log(cp(n))$ bits. Finally, at the leaf, the decision tree outputs $\Enc(w^\star)_i$ where $i$ is the index in $\{1,\ldots,cp(n)\}$ defined by the last $\log(cp(n))$ bits of the input. It follows that $T_z$ exactly computes $\Cert_z$. Additionally, the size of the decision tree is $n+cp(n)$ as desired.

    If $w^\star$ is given as input, then $\Enc(w^\star)$ can be computed in polynomial time. Therefore, there is an efficient algorithm which outputs $T_z$. 
\end{proof}

In the proof of \Cref{claim:enumerable}, we will use an oracle for an augmented variant of the language of a verifier $V$, called $\textsc{Lex}(V)$, to help us find lexicographically first certificates. 

\begin{definition}[The language $\textsc{Lex}(V)$]
    Let $V$ be a verifier, and for a string $w\in\zo^n$, let $\textsc{Rank}(w)\in [2^n]$ denote the rank of $w$ in the lexicographic ordering of strings in $\zo^{n}$. 
    Then, $\textsc{Lex}(V)$ is the language defined as
    $$
    \textsc{Lex}(V)=\{(x,k)\mid \text{there is a certificate }w\text{ s.t. } V(x,w)=1\text{ and }\textsc{Rank}(w)\le k\}.
    $$
\end{definition}

\begin{proposition}[Verifiers for the language $\textsc{Lex}(V)$]
\label{prop:lex}
    Let $V$ be a time-$O(t(n))$ verifier using size-$O(p(n))$ certificates for growth functions $t(n),p(n)\ge n$. Then, we have that $\textsc{Lex}(V)\in \NTIME(t(n), p(n))$.
\end{proposition}

\begin{proof}
    We construct a verifier $V_{\mathrm{lex}}$ for the language $\textsc{Lex}(V)$ which runs in time $O(p(n))+t(n)$ and uses certificates of size $p(n)$. On input $(x,k)$ where $x\in \zo^n$ and $k\in \zo^{p(n)}$ and given certificate $w\in \zo^{p(n)}$, $V_{\mathrm{lex}}$ checks that $w$ is within the lexicographically first $k$ strings in $\zo^{p(n)}$ and if so, outputs $V(x,w)$. The correctness of the verifier follows immediately from the correctness of $V$. Moreover, the runtime is $O(p(n))+t(n)$ and, since $p(n)\le t(n)$, this is $O(t(n))$. Note that the length of the input is $N=n+p(n)$. Therefore, the runtime is still $O(t(n))$ and the certificate length is at most $p(N)$ by our assumption that $v,p$ are both growth functions. 
\end{proof}

\begin{proof}[Proof of \Cref{claim:enumerable}]
    First, our algorithm uses an oracle for $\textsc{Lex}(V)$ to find the lexicographically first certificate $w^\star$ such that $V(z,w^\star)=1$, if such a certificate exists. This certificate can be found using binary search and requires at most $p(n)$ calls to the oracle. If no such certificate exists, then the algorithm outputs the decision tree for the constant $0$ function and terminates. Otherwise, the algorithm obtains $w^\star$, and using the algorithm from \Cref{prop:cert_dt} outputs the decision tree $T_z$. 

    \paragraph{Runtime.}{
        The algorithm requires $O(n)+p(n)$ time to find $w^\star$. Then it requires $\poly(p(n))$ time to compute $T_z$. So the overall runtime is $\poly(p(n))$ as desired.
    }
    \paragraph{Oracle complexity.}{
        By \Cref{prop:lex}, the language $\textsc{Lex}(V)$ is in $\NTIME(t(n),p(n))$ and therefore our algorithm makes $p(n)$ oracle calls to a language in $\NTIME(t(n),p(n))$ as desired.\hfill\qedhere
    }
\end{proof}

\subsubsection{\Cref{thm:NP formal:many samples}: The concept class can be learned quickly with enough samples}

In this section, we show that $\mathcal{C}_{L,V}$ can be efficiently learned with enough samples. This fact follows from a more general one saying that concept classes consisting of \textit{sparse functions} can be quickly learned with sufficiently many samples.

\begin{claim}[Sparse concept classes can be learned quickly with many samples]
\label{claim:large-samples-upper-bound}
  Let $\mcC$ be a concept class and $p(n)\ge n$ be a growth function such that every function $f\in\mcC$ is $p(n)$-sparse. Then, for every $\eps>0$, there is an algorithm which learns $\mcC$ to error $\eps$ and whose runtime and sample complexity are both $O(p(n)/\eps)$.
\end{claim}

\begin{proof}
    Fix an unknown distribution $\mathcal{D}$ over examples and a target concept $f:\zo^n\to\zo$. Consider the following empirical risk minimizer using $m$ samples to learn for some $m$ which we specify later. First, it draws a random sample $\bS$ of size $m$. Then, it outputs the hypothesis $h_{\bS}$ defined as
    $$
    h_{\bS}(x)=\begin{cases}
        1 & \text{if }(x,1)\in\bS\\
        0 & \text{else}
    \end{cases}.
    $$
    We will show that with high probability over $\bS$, the hypothesis $h_{\bS}$ satisfies $\error_{\D}(h_{\bS}, f)\le \eps$. 
    
    Let $\mathcal{H}$ denote the class of all possible $h_{S}$ for some size-$m$ sample $S$. Since $h_S$ is uniquely defined by the $1$-inputs in $S$, and there are $p(n)$ many $1$-inputs, the size of $\mathcal{H}$ is at most $2^{p(n)}$. Therefore, we can write
    \begin{align*}
        \Prx_{\bS}[\text{some }h\in \mathcal{H}&\text{ is consistent with }f\text{ over }\bS\text{ and } \error_{\D}(f,h)>\eps]\\
        &\le \sum_{h\in \mathcal{H}}\Prx_{\bS}[h\text{ is consistent with }f\text{ over }\bS\text{ and }\error_{\D}(f,h)>\eps]\tag{Union bound}\\
        &\le |\mathcal{H}|(1-\eps)^m\le |\mathcal{H}|\exp(-\eps m).\tag{Samples are i.i.d.}
    \end{align*}
    By choosing $m=O(\log |\mathcal{H}|/\eps)=O(p(n)/\eps)$, we can ensure that with probability at least $0.99$, the $h$ which our algorithm outputs satisfies $\error_{\D}(f,h)\le \eps$. Moreover, the runtime of this algorithm is $O(p(n)/\eps)$ and the claim is proved. 
\end{proof}

Since the functions in $\mathcal{C}_{L,V}$ are $O(p(n))$-sparse, \Cref{claim:large-samples-upper-bound} implies that $\mathcal{C}_{L,V}$ can be learned quickly with sufficiently many samples:

\begin{corollary}[$\mathcal{C}_{L,V}$ can be learned quickly with many samples]
\label{cor:large-samples-upper-bound}
Let $L\in \NTIME(t(n),p(n))$ for growth functions $v(n),p(n)\ge n$ and let $V$ be a time-$O(t(n))$ verifier for $L$ using size-$O(p(n))$ certificates. Then, for every $\eps>0$, there is an algorithm which learns $\mathcal{C}_{L,V}$ to error $\eps$ whose runtime and sample complexity are both $O(p(n)/\eps)$.
\end{corollary}






\section{Main reduction: Proof of~\Cref{thm:NP formal:few samples}}

\subsection{$\mcC_{L,V}$ inherits its time versus sample complexity tradeoffs from the time versus nondeterminism tradeoffs for $L$}

The focus of this section will be to prove the following lemma. 

\begin{lemma}
\label{lem:one curve lower bounds the other} 
Let $L\in \NTIME(t(n),p(n))$ for growth functions $t(n),p(n)\ge n$, let $V$ be a time-$O(t(n))$ verifier for $L$ using size-$O(p(n))$ certificates, and let $\eps^\star>0$ be the constant from \Cref{def:codes}. Suppose that $L\not\in \AMTIME(\ell(n),m(n))$. Then, any $m(n)$-sample algorithm which learns $\mcC_{L,V}$ to error $\eps^\star$ must run in time $\Omega\left(\frac{\ell(n)-t(n)}{p(n)\log \ell(n)}-\poly(p(n))\right)$.    
\end{lemma}

\paragraph{Example setting of parameters in \Cref{lem:one curve lower bounds the other}.}{
We typically think of $m(n)\ll p(n)$ and $\ell(n)\gg t(n)\ge p(n)$. This corresponds to the intuition that as one shrinks the nondeterminism budget, the time complexity should increase. In the case of SAT, we have $t(n)=n\log n$ and $p(n)=n$. So, for example, if $\SAT \notin \AMTIME(2^{\sqrt{n}},\sqrt{n})$, i.e.~if any $\AM$ protocol for $\SAT$ that uses $O(\sqrt{n})$ bits of nondeterminism requires $\Omega(2^{\sqrt{n}})$ time, then \Cref{lem:one curve lower bounds the other} says that any $O(\sqrt{n})$-sample learner for $\mathcal{C}_{\mathrm{SAT}}$ must run in time $\Tilde{\Omega}(2^{\sqrt{n}})$.
}



\subsection{Proof of~\Cref{lem:one curve lower bounds the other}}

The proof of \Cref{lem:one curve lower bounds the other} is algorithmic and proceeds via the contrapositive. The key step is a randomized, nondeterministic reduction from verifying $L$ with few nondeterministic bits to learning $C_{L,V}$ with few samples.

\begin{lemma}[Learners with few samples yield $\AM$ protocols with few bits of nondeterminism]
\label{claim:nondeterministic-reduction}
    Let $L\in \NTIME(t(n),p(n))$ for growth functions $t(n),p(n)\ge n$, $V$ be a time-$O(t(n))$ verifier for $L$ using size-$O(p(n))$ certificates, and $\eps^\star>0$ be the constant from \Cref{def:codes}. If there is a learner for $\mathcal{C}_{L,V}$ which uses $m(n)$ samples, runs in time $T(n)$, and outputs a hypothesis with error $\eps^\star$ then, $$L\in\AMTIME(p(n)T(n)\log T(n)+\poly(p(n))+t(n), m(n)).$$ Furthermore, the $\AM$ protocol for $L$ has perfect soundness. 
\end{lemma}

\begin{proof}
    Let $A$ be the learner for $\mathcal{C}_{L,V}$. We design an $\AM$~protocol $M$ for $L$. On input $z\in \zo^n$, $M$ does the following.
    \begin{enumerate}
        \item Draw $m=m(n)$ uniform random indices $\bi^{(1)},\ldots,\bi^{(m)}\sim \zo^{\log (cp(n))}$.
        \item Ask Merlin for the $m$ labels, $w_1,\ldots,w_m$, of the examples $\bx^{(1)},\ldots,\bx^{(m)}$ where $\bx^{(j)}=(z,\bi^{(j)})$
        \item Run the learner $A$ on the sample $\bS=\{(\bx^{(1)},w_1),\ldots,(\bx^{(m)},w_m)\}$ and obtain a hypothesis $\textsc{Hyp}:\zo^{n+\log (C p(n))}\to \zo$.
        \item Query \textsc{Hyp} on all points $(z,i)$ for $i\in\zo^{\log (cp(n))}$ to obtain a string $y\in \zo^{cp(n)}$.
        \item Apply the decoding algorithm to $y$ to obtain $\Tilde{w}=\Dec(y)$.
        \item Run the verifier $V$ on $z$ using certificate $\Tilde{w}$ and output the result, i.e. output $V(z,\Tilde{w})$.
    \end{enumerate}

    \paragraph{Runtime.}{
        The learner $A$ runs in time $T(n)$ and therefore the hypothesis that it outputs can simulated on a query in time $O(T(n)\log T(n))$. So obtaining $y$ takes time $O(p(n)T(n)\log T(n))$. The decoding algorithm runs in time $\poly(p(n))$. Finally, running the verifier $V$ takes time $t(n)$. So the overall runtime of $M$ is $O(p(n)T(n)\log T(n)+\poly(p(n))+t(n))$.  
    }
    \paragraph{Correctness.}{
        We prove the completeness and soundness of this protocol separately.

        \subparagraph{Completeness.}{
            If $z\in L$, then for every choice of examples $x^{(1)},\ldots,x^{(m)}$, Merlin can respond with the true labels given by $\Cert_z$. It follows that $\bS$ is a size-$m$ sample from the distribution $\mathcal{D}$ which is uniform over the set $\{(z,i)\mid i\in \zo^{cp(n)}\}$. Therefore, w.h.p.~$\textsc{Hyp}$ has error $\eps^\star$ over $\mathcal{D}$ and the string $y$ contains at most $\eps^\star\cdot cp(n)$ errors relative to $\Enc(w^\star)$. The decoder then returns $w^\star$. Finally, since $w^\star$ is a certificate for $z$, the output of $M$ is $1=V(z,w^\star)$. 
        }
        \subparagraph{Soundness.}{
            If $z\not\in L$, then  $V(z,\Tilde{w})=0$ for every certificate $\Tilde{w}$. Therefore, $M$ always outputs the correct answer regardless of Merlin's responses.\hfill\qedhere
        }
    }
\end{proof}

\Cref{lem:one curve lower bounds the other} now follows from the contrapositive of \Cref{claim:nondeterministic-reduction}.


With~\Cref{claim:nondeterministic-reduction} in hand, we can then derive a one-sided randomized algorithm for $L$:

\begin{corollary}[Learners with few samples yield fast randomized algorithms]
    \label{cor:randomized reduction}
    Let $L\in \NTIME(t(n),p(n))$ and let $V$ be a time-$O(t(n))$ verifier for $L$ using size-$O(p(n))$ certificates. There is a constant $\eps>0$ such that if there is a learner for $\mathcal{C}_{L,V}$ which uses $m(n)$ samples, runs in time $T(n)$, and outputs a hypothesis with error $\eps$ then, $L\in\RTIME(2^{O(m(n))}T(n)\log T(n)\poly(t(n)))$.
\end{corollary}

This follows from the fact that $\AM$ protocols with perfect soundness can be simulated by randomized algorithms with one-sided error:

\begin{proposition}[Randomized simulation of $\AMTIME$]
    \label{prop:amtime upper bounds}
    If $L\in \AMTIME(t(n),p(n))$ and $L$ is realized by an $\AM$~protocol with perfect soundness, then $L\in \RTIME(t(n)2^{O(p(n))})$.
\end{proposition}

\begin{proof}
    Let $M$ be an $\AM$~protocol with perfect soundness which runs in time $O(t(n))$ and uses $O(p(n))$ bits of nondeterminism. We design a randomized algorithm $A$ for $L$. On input $z\in \zo^n$, $A$ does the following.
    \begin{enumerate}
        \item Generate a random string $\br$.
        \item Output 1 if there is some $w\in\zo^{O(p(n))}$ such that $M(\br,w,z)=1$.
        \item Otherwise, output $0$.
    \end{enumerate}
    Since $A$ simulates $M$ on every $w\in\zo^{O(p(n))}$, the runtime is $O(t(n)) \cdot 2^{O(p(n))}$. 
    
    It remains to show this algorithm has one-sided error. If $z\in L$, then by the completeness of $M$, with probability $2/3$ over the choice of $\br$, there will be some $w\in\zo^{O(p(n))}$ on which $M(\br,w,z)=1$ in which case $A$ successfully outputs $1$. On the other hand, if $z\not\in L$, then by perfect soundness, \textit{for every} $\br$ and \textit{for every} $w\in\zo^{O(p(n))}$, we have $M(\br,w,z)=0$ in which case our algorithm correctly outputs $0$. 
\end{proof}

\begin{proof}[Proof of \Cref{cor:randomized reduction}]
    By \Cref{claim:nondeterministic-reduction}, since $p(n)\le t(n)$, $L\in \AMTIME(T(n)\log T(n)\poly(t(n)), m(n))$ and this is realized with an $\AM$ protocol with perfect soundness. The corollary then follows from \Cref{prop:amtime upper bounds} which shows that: $$\AMTIME(T(n)\log T(n)\poly(t(n)), m(n))\sse \RTIME(2^{O(m(n))}T(n)\log T(n)\poly(t(n)))$$ and therefore $L\in \RTIME(2^{O(m(n))}T(n)\log T(n)\poly(t(n)))$ as desired.
\end{proof}

\section{Proof of \Cref{thm:NP formal} and its corollaries}

We are now ready to put things together and prove~\Cref{thm:NP formal}. Let $L^\star$ be the language
    $$
    L^\star \coloneqq \{\langle{M,x}\rangle\mid M\text{ accepts }x\text{ in }p(|x|)\text{ steps}\}
    $$
    where $M$ is the encoding a nondeterministic Turing machine.
    Let $V^\star$ be the verifier that on input $\langle M,x\rangle$ and certificate $w$, simulates $M$ on $(x,w)$ for $p(|x|)$ time steps and outputs the value of $M(x,w)$. Since time-$p(n)$ Turing machines can be simulated in time $O(p(n)\log p(n))$ \cite{AB09}, this verifier uses $p(n)$-size certificates, runs in $O(p(n)\log p(n))$ time, and verifies $L^\star$, thereby showing that $L^\star\in \NTIME(p(n)\log p(n), p(n))$. 
    
    We claim that the concept class $\mcC^\star\coloneqq \mcC_{L^\star,V^\star}$ fulfills the requirements of the theorem:
    \begin{enumerate}
        \item By \Cref{claim:vc-dimension}, $\VCdim(\mcC^\star)=1$ and by \Cref{claim:enumerable}, the class is $\NTIME(p(n)\log(p(n)),p(n))$-enumerable. Also, by \Cref{claim:vc-dimension}, there is an algorithm that learns $\mcC^\star$ to error $\eps$ for every $\eps>0$ using $O(1/\eps)$ samples in time $2^{O(p(n))}$. 
        \item By \Cref{cor:randomized reduction}, if there is an $m$-sample, time-$t(n)$ algorithm that learns $\mcC^\star$ w.h.p.~to error $\eps^\star$, then $L^\star\in \RTIME\left(2^{O(m(n))}t(n)\log t(n)\poly(p(n))\right)$. We also observe that $L^\star$ is $\NTIME(p(n))$-hard. Given a language $L\in \NTIME(p(n))$ with verifier $V$, we can reduce it to $L^\star$ via the mapping $x\mapsto \langle V,x\rangle$. Correctness is immediate from the definition of $L^\star$. It follows that if $L^\star\in \RTIME\left(2^{O(m(n))}t(n)\poly(p(n))\right)$, then $\NTIME(p(n))\sse \RTIME\left(2^{O(m(n))}t(n)\log t(n)\poly(p(n))\right)$.
        \item By \Cref{claim:large-samples-upper-bound}, the class $\mcC^\star$ can be learned to error $\eps$ for every $\eps>0$ in time $O(p(n)/\eps)$ using $O(p(n)/\eps)$ samples.
    \end{enumerate}
    Finally, by \Cref{prop:cert_dt}, every concept in $\mcC^\star$ is a decision tree of size $O(p(n))$. \hfill\qed

\subsection{Corollaries of \Cref{thm:NP formal}}
\begin{corollary}[Consequences under $\NP\neq\RP$]
\label{cor:np-neq-rp}
    There is a concept class $\mathcal{C}$ such that the following holds. 
    \begin{enumerate}
        \item We have $\VCdim(\mcC)=1$ and $\mathcal{C}$ is $\NP$-enumerable. There is an algorithm that, for every $\eps>0$, learns $\mcC$ to error $\eps$ using $O(1/\eps)$ samples in time $2^{O(n)}$. 
        \item Let $\eps^\star>0$ be the constant in \Cref{def:codes}. Assuming $\NP\neq \RP$, no polynomial-time algorithm can learn $\mcC$ to error $\eps^\star$ using $O(\log n)$ samples.  
        \item There is an algorithm that, for every $\eps>0$, learns $\mcC$ to error $\eps$ whose sample complexity and runtime are both $O(n/\eps)$. 
    \end{enumerate}
\end{corollary}

\Cref{cor:np-neq-rp} implies \Cref{cor:dichotomy} since $O(\log n)$ is asymptotically larger than any function of $\VCdim(\mcC)=1$. 

\begin{proof}
    We apply \Cref{thm:NP formal} with $p(n)=n$ and $m(n)=O(\log n)$. To prove the second part of the corollary, suppose for contradiction that there is a polynomial-time, $O(\log n)$-sample algorithm that learns $\mcC$ to error $\eps$ where $\eps>0$ is the constant from \Cref{thm:NP formal}. Then, \Cref{thm:NP formal} implies that $\NTIME(n)\sse \RTIME(\poly(n))$. By padding, this implies that $\NP\sse\RP$, contradicting our assumption.
\end{proof}

\begin{corollary}[Consequences under randomized ETH]
    There is a concept class $\mathcal{C}$ such that the following holds. 
    \begin{enumerate}
        \item We have $\VCdim(\mcC)=1$ and $\mathcal{C}$ is $\NP$-enumerable. There is an algorithm that, for every $\eps>0$, learns $\mcC$ to error $\eps$ using $O(1/\eps)$ samples in time $2^{O(n)}$. 
        \item Let $\eps^\star>0$ be the constant in \Cref{def:codes}. Assuming randomized ETH, there is a constant $\delta>0$ such that, no $2^{\delta n}$-time algorithm can learn $\mcC$ to error $\eps^\star$ using $\delta n$ samples.
        \item There is an algorithm that, for every $\eps>0$, learns $\mcC$ to error $\eps$ whose sample complexity and runtime are both ${O}(n/\eps)$. 
    \end{enumerate}
\end{corollary}
   
\begin{proof}
    Let $m(n)=\delta n$ for a small constant $\delta>0$ specified later. The language $\mathrm{SAT}$ is in $\NTIME(n\polylog n, n)$. Therefore, the class $\mcC_{\text{SAT}}$ satisfies $\VCdim(\mcC_{\text{SAT}})=1$ and is $\NP$-enumerable. Also, by \Cref{cor:large-samples-upper-bound}, there is an algorithm that, for every $\eps>0$, learns $\mcC_{\text{SAT}}$ to error $\eps$ whose sample complexity and runtime are both ${O}(n/\eps)$. It remains to prove the second point of the corollary. By \Cref{claim:nondeterministic-reduction}, if there is a $2^{\delta n}$-time, $\delta n$-sample algorithm that learns $\mcC_{\text{SAT}}$ to error $\eps^\star$, then $\mathrm{SAT}\in \AMTIME(2^{O(\delta n)}, \delta n)$, and \Cref{prop:amtime upper bounds} implies that $\mathrm{SAT}\in \RTIME(2^{O(\delta n)})$. This contradicts ETH for sufficiently small $\delta>0$. 
\end{proof}

\begin{corollary}[Consequences under exponential-time lower bounds for $\NP$]
    Let $p(n)$ be a polynomial and suppose there is a language $L\in \NTIME(p(n))$ such that any randomized algorithm for $L$ requires time $2^{\Omega(p(n))}$. Then, there is a concept class $\mathcal{C}$ such that the following holds.
    \begin{enumerate}
        \item We have $\VCdim(\mcC)=1$ and $\mathcal{C}$ is $\NP$-enumerable. There is an algorithm that, for every $\eps>0$, learns $\mcC$ to error $\eps$ using $O(1/\eps)$ samples in time $2^{O(p(n))}$. 
        \item Let $\eps^\star>0$ be the constant in \Cref{def:codes}. There is a constant $\delta>0$ such that no algorithm running in time $2^{\delta p(n)}$ can learn $\mathcal{C}$ to error $\eps^\star$ using $\delta  p(n)$ samples.
        \item There is an algorithm that, for every $\eps>0$, learns $\mcC$ to error $\eps$ whose sample complexity and runtime are both $O(p(n)/\eps)$.
    \end{enumerate}
\end{corollary}

\begin{proof}
    We apply \Cref{thm:NP formal} with $p(n)$ and $m(n)=\delta p(n)$ for a small constant $\delta>0$ specified later. To prove the second part of the corollary, suppose for contradiction that there is a $2^{\delta p(n)}$-time algorithm that can learn $\mcC$ to error $\eps^\star$ using $\delta p(n)$ samples. Then, \Cref{thm:NP formal} implies that $\NTIME(p(n))\sse \RTIME(2^{O(\delta p(n))})$ which, for sufficiently small $\delta$, contradicts our randomized time lower bound for $L$.
\end{proof}

\section{Extensions to other settings}

\subsection{Uniform Distribution}

In this section, we extend \Cref{thm:NP formal} to hold for \textit{uniform-distribution} learning. 

\begin{theorem}[Extension of \Cref{thm:NP formal} to the uniform distribution]
\label{thm:uniform}
    For every time constructible growth function $p(n)\ge n$, there is a concept class $\mcC$ over $\zo^n$ such that the following holds.     
    \begin{enumerate}
        \item We have $\VCdim(\mcC)\le n$ and $\mcC$ is $\NTIME(p(n)\log p(n), p(n))$-enumerable. There is an algorithm that, for every $\eps>0$, learns $\mcC$ to error $\eps$ using $O(n/\eps)$ samples in time $2^{O(p(n))}$. 
        \item Let $\eps^\star$ be the constant in \Cref{def:codes}. If there is an $m(n)$-sample, time-$t(n)$ algorithm that learns $\mcC$ to error $\eps^\star/100$ \emph{over the uniform distribution}, then $$\NTIME(p(n))\sse \RTIME\left(2^{O(m(n))}t(n)\log t(n)\poly(p(n))\right).$$
        \item There is an algorithm that for every $\eps>0$ learns $\mcC$ to error $\eps$ whose sample complexity and runtime are both $O(p(n)/\eps)$.
    \end{enumerate}
\end{theorem}

For example, under the assumption that there is a language in $\NTIME(n^{100})$ that requires $2^{\Omega(n^{100})}$ time on randomized Turing machines, \Cref{thm:uniform} gives a concept class that,
\begin{enumerate}
    \item[(i)] Can be learned inefficiently with $O(n)$ samples, over any distribution.
    \item[(ii)] Cannot be learned efficiently with $c_1 \cdot n^{100}$ samples for an appropriate constant $c_1$ even over just the uniform distribution.
    \item[(iii)] Can be learned efficiently with $c_2 \cdot n^{100}$ samples for appropriate $c_2 > c_1$ over any distribution.
\end{enumerate}

The main difference between \Cref{thm:uniform} and \Cref{thm:NP formal} is that the collapse in part (ii) of \Cref{thm:uniform} holds even if the algorithm only learns $\mcC$ over the uniform distribution. This comes at the cost of increasing the VC dimension of the concept class. 

Recall that in the proof of \Cref{thm:NP formal}, the concepts in $\mcC$ are \textit{sparse}. Every concept $c\in\mcC$ is highly biased toward $0$. This makes it easy to learn $\mcC$ over the uniform distribution with $0$ samples. So, to achieve the collapse in \Cref{thm:uniform} for the uniform distribution, we will have to modify the concept class. Additionally, the concepts $c\in\mcC$ from \Cref{thm:NP formal} have some inputs which are ``useless'' and some which are ``useful'' (recall \Cref{def:concept-class}). The distribution we construct to prove part (ii) of \Cref{thm:NP formal} is uniform over useful examples. To generalize to the uniform distribution, we will want every example to provide the same amount of information as all the others. Given these considerations, we define the following variant of the concept class in \Cref{def:concept-class}.

\begin{definition}[Definition of the concept class $\mathcal{U}_{L,V}$]
    \label{def:uniform}
    Let $L\in \NTIME(t(n),p(n))$ for growth functions $t(n),p(n)\ge n$ and let $V$ be a verifier for $L$ running in time $t(n)$ whose certificate size is bounded by $p(n)$. Let $(\Enc,\Dec)$ be the error-correcting code from \Cref{def:codes} with rate $1/c$ for constant $c>1$. For $z\in \zo^n$, we define a concept $\UC_z:\zo^{\log (cp(n))+n}\to \zo$ as follows. If $z\not\in L$, then $\UC_z$ is the constant $0$ function. Otherwise, let $w^\star\in \{0,1\}^{p(n)}$ be the lexicographically first certificate such that $V(z,w^\star)=1$ and define $\UC_z$ as 
    $$
    \UC_z(i,x)=\Enc\left(w^\star\right)_i
    $$
    where $i\in \zo^{\log (cp(n))}$ is interpreted as an integer index in the range $\{1,\ldots,cp(n)\}$. Finally, the concept class $\mathcal{U}_{L,V}$ is defined to be the set of all $\UC_z$ for $z\in\zo^n$. 
\end{definition}

It is straightforward to verify the analogues of \Cref{claim:vc-dimension}, \Cref{cor:large-samples-upper-bound,claim:enumerable} for the class $\mathcal{U}_{L,V}$:
\begin{enumerate}
    \item There are at most $2^n$ distinct concepts $\UC_z$ and so $\VCdim(\mathcal{U}_{L,V})\le n$ by \Cref{fact:vc-upper}. Also, following the proof of \Cref{claim:enumerable}, $\mathcal{U}_{L,V}$ is $\NTIME(t(n),p(n))$-enumerable. Therefore, by \Cref{fact:consistent}, there is an algorithm that, for every $\eps>0$, learns $\mcC$ to error $\eps$ using $O(n/\eps)$ samples in time $2^{O(p(n))}\cdot O(t(n))$.
    \item Each concept $\UC_z$ is a $\log(cp(n))$-junta over the first $\log(cp(n))$ bits of input. In particular, the junta is trivially $O(p(n))$-sparse and therefore can be learned to error $\eps$ via \Cref{claim:large-samples-upper-bound} using time and sample complexity $O(p(n)/\eps)$. This gives an $\eps$-error hypothesis over the marginal distribution of $\bi$ where $(\bi,\bx)\sim\D$. By evaluating the junta hypothesis on the first $\log(cp(n))$ bits of an input, we obtain an $\eps$-error hypothesis for $\UC_z$ over $\D$.
\end{enumerate}

It remains to prove the extension of \Cref{cor:randomized reduction} to learning algorithms for $\mathcal{U}_{L,V}$ over the uniform distribution.
\begin{claim}[Uniform distribution extension of \Cref{cor:randomized reduction}]
\label{claim:uniform-dist-reduction}
    Let $L\in \NTIME(t(n),p(n))$, let $V$ be a time-$O(t(n))$ verifier using size-$O(p(n))$ certificates, and let $\eps^\star$ be the constant from \Cref{def:codes}. If there is a algorithm for learning $\mathcal{U}_{L,V}$ over the uniform distribution which uses $m(n)$ samples, runs in time $T(n)$, and outputs a hypothesis with error $\eps^\star/100$, then we have that $L\in\RTIME(2^{O(m(n))}T(n)\log T(n)\poly(t(n)))$.
\end{claim}

\begin{proof}
    Let $A$ be the learner for $\mathcal{C}_{L,V}$. Let $\eps^\star>0$ be the constant from \Cref{def:codes} which represents the fraction of errors our decoding algorithm can decode. Following the proof of \Cref{claim:nondeterministic-reduction}, we first design a randomized $\AM$~protocol $M$ for $L$. On input $z\in \zo^n$, $M$ does the following.
    \begin{enumerate}
        \item Draw $m=m(n)$ uniform random examples $\by^{(1)},\ldots,\by^{(m)}$ where $\by^{(j)}=(\bi^{(j)},\bx^{(j)})$.
        \item Ask Merlin for the $m$ labels, $w_1,\ldots,w_m$, of the examples $\by^{(1)},\ldots,\by^{(m)}$.
        \item Run the learner $A$ on the sample $\bS=\{(\bx^{(1)},w_1),\ldots,(\bx^{(m)},w_m)\}$ with error parameter $\eps/100$ and obtain a hypothesis $\textsc{Hyp}:\zo^{n+\log (cp(n))}\to \zo$.
        \item Pick a uniform random $\bx\sim\zo^n$ and query \textsc{Hyp} on all points $(i,\bx)$ for $i\in\zo^{\log (cp(n))}$ to obtain a string $y\in \zo^{cp(n)}$ which is a noisy version of $\Enc(w^\star)$.
        \item Apply the decoding algorithm to $y$ to obtain $\Tilde{w}=\Dec(y)$.
        \item Run the verifier $V$ on $z$ using certificate $\Tilde{w}$ and output the result, i.e. output $V(z,\Tilde{w})$.
    \end{enumerate}
    \paragraph{Runtime.}
    The learner $A$ runs in time $T(n)$ and therefore the hypothesis that it outputs can simulated on a query in time $O(T(n)\log T(n))$. So obtaining $y$ takes time $O(p(n)T(n)\log T(n))$. The decoding algorithm runs in time $\poly(p(n))$. Finally, running the verifier $V$ takes time $t(n)$. So the overall runtime of $M$ is $O(p(n)T(n)\log T(n)+\poly(p(n))+t(n)) = O(T(n)\log T(n)\poly(t(n)))$ since $p(n)\le t(n)$. 
    \paragraph{Correctness.} We prove the completeness and soundness of this protocol separately.

    \subparagraph{Completeness.} If $z\in L$, then for every choice of examples $y^{(1)},\ldots,y^{(m)}$, Merlin can respond with the true labels given by $\UC_z$. The sample $\bS$ is then a size-$m$ sample from the uniform distribution and with probability $0.99$, the learner outputs a hypothesis $\textsc{Hyp}$ satisfying $\error(\textsc{Hyp}, \UC_z)\le \eps^\star/100$. We can rewrite this error as
    $$
    \Ex_{\bx\sim\zo^n}\left[\Prx_{\bi\sim\zo^{\log (cp(n))}}[\textsc{Hyp}(\bi,\bx)\neq \UC_z(\bi,\bx)] \right]\le \frac{\eps}{100}
    $$
    which, by Markov's inequality, implies that with probability at least $0.99$ over a random $\bx\sim\zo^n$, we have
    $$
    \Prx_{\bi\sim\zo^{\log (cp(n))}}[\textsc{Hyp}(\bi,\bx)\neq \UC_z(\bi,\bx)]\le \eps^\star.
    $$
    Therefore, with high probability, the string $y$ has at most $\eps^\star\cdot cp(n)$ errors relative to $\Enc(w^\star)$. The decoder then returns $w^\star$. Finally, since $w^\star$ is a certificate for $z$, the output of $M$ is $1=V(z,w^\star)$. 

    \subparagraph{Soundness.}{
            If $z\not\in L$, then for every certificate $\Tilde{w}$, $V(z,\Tilde{w})=0$. Therefore, $M$ always outputs the correct answer regardless of Merlin's responses.
        }
    Since $M$ is an $\AM$ protocol for $L$ with perfect soundness, uses $m(n)$ bits of nondeterminism, and runs in time $O(T(n)\log T(n)\poly(t(n)))$, then by \Cref{prop:amtime upper bounds}, we get $L\in\RTIME(2^{O(m(n))}T(n)\log T(n)\poly(t(n)))$ as desired.
\end{proof}

We can now combine our observations to give a proof of \Cref{thm:uniform}.

\begin{proof}[Proof of \Cref{thm:uniform}]
    This proof is identical to that of \Cref{thm:NP formal} except that we replace the concept class $\mcC^\star$ with $\mathcal{U}^\star\coloneqq \mathcal{U}_{L^\star,V^\star}$ where $L^\star$ and $V^\star$ are as in \Cref{thm:NP formal}. As we have already observed, $\mathcal{U}^\star$ has $\VCdim(\mathcal{U}^\star)\le n$ and is $\NTIME(p(n)\log p(n),p(n))$-enumerable and there is an algorithm that, for every $\eps>0$, learns $\mcC$ to error $\eps$ using $O(n/\eps)$ samples in time $2^{O(p(n))}$. Also, it can be learned to error $\eps$ over any distribution with time and sample complexity $O(p(n)/\eps)$. Finally, point (ii) of the theorem follows from \Cref{claim:uniform-dist-reduction} and the fact that $L^\star$ is $\NTIME(p(n))$-hard.
\end{proof}

\subsection{Online Learning}

In this section, we extend \Cref{thm:NP formal} to the setting of online learning. An efficient online learner is one that must label examples efficiently; in other words, we bound the algorithm's runtime for each round of online learning. This quantity is analogous to the overall time complexity of a PAC learning algorithm. 
Similarly, the mistake bound in online learning is analogous to sample complexity in the PAC model. Thus, the analogue of time-sample tradeoffs in the PAC model is time-mistake tradeoffs in the online learning model.

\begin{definition} [Littlestone dimension]\label{def:littlestone}
The \emph{Littlestone dimension} of a concept class $\mcC$ is the largest number $M$ such that every online learning algorithm for $\mcC$ must make at least $M$ mistakes. We denote this quantity $\Ldim(\mcC)$.
\end{definition}

The standard definition of Littlestone dimension is combinatorial in nature. However, we only use the optimal mistake bound definition and it is a well known fact that the two are equivalent~\cite{Lit88}.

\begin{theorem}[\Cref{thm:NP formal} for Online Learning]
\label{thm:online}
     For every time constructible growth function $p(n)\ge n$, there is a concept class $\mcC$ such that the following holds. 
    \begin{enumerate}
        \item We have $\Ldim(\mcC)=1$ and $\mcC$ is $\NTIME(p(n)\log p(n),p(n))$-enumerable. Consequently, there is an algorithm that online-learns $\mcC$ with at most 1 mistake and takes time $2^{O(p(n))}$ in each round.\label{thm:online:ldim}  
        \item If there is an online learning algorithm that makes at most $m(n)$ mistakes and takes time $t(n)$ in each round, then $\NTIME(p(n))\sse \RTIME\left(2^{O(m(n))}T(n)\log(T(n))\poly(p(n))\right)$ where $T(n)=t(n)m(n)\log(m(n))$. \label{thm:online:few mistake}  
        \item There is an algorithm  that online-learns $\mcC$ with at most $O(p(n))$ mistakes and takes time at most $O(n\log(p(n))+ \log^2(p(n))$ in each round.\label{thm:online:many mistake}  
    \end{enumerate}

\end{theorem}

Note that when setting parameters, we generally have $m(n)=\log(t(n))$. Therefore,~\Cref{thm:online:few mistake} and~\Cref{thm:NP formal:few samples} achieve the same parameters save extra $\log$-factors. Since we set $t(n)$ to be polynomial or even exponential in $n$ to achieve~\Cref{thm:NP formal}'s relevant corollaries, this log factor can safely be ignored.  

\paragraph{Ingredients for the proof of~\Cref{thm:online}.} To prove \Cref{thm:online}, we will use the same concept class as in the proof of~\Cref{thm:NP formal} (see~\Cref{def:concept-class}). We must prove the two upper bounds ~\Cref{thm:online:ldim} and~\Cref{thm:online:many mistake} essentially from scratch. However, the analogous bounds in the PAC setting,~\Cref{thm:NP formal:VCdim} and~\Cref{thm:NP formal:many samples}, were relatively easy to show, and the proofs are even simpler in the online model. The two claims~\Cref{claim:littlestone-dimension} and~\Cref{claim:many-mistakes-upper-bound} will establish these two upper bounds. 

On the other hand, we can use the more complicated lower bound~\Cref{thm:NP formal:few samples} as a black box to establish~\Cref{thm:online:few mistake}. This follows from a classical result~\Cref{lem:online lower bound}~\cite{Lit89}
that online learning is at least as hard as PAC learning. 


\begin{claim}[\Cref{claim:vc-dimension} for Online Learning]
\label{claim:littlestone-dimension}
    Let $L\in \NTIME(t(n),p(n))$ and let $V$ be a time-$O(t(n))$ verifier for $L$  using size-$O(p(n))$ certificates. Then, $\Ldim(\mathcal{C}_{L,V}) = 1$. 
\end{claim}
\begin{proof}
Since $\Ldim(\mathcal{C})\geq \VCdim(\mathcal{C})$ for all concept classes $\mathcal{C}$, it follows from~\Cref{claim:vc-dimension} that $\Ldim(\mathcal{C}_{L,V}) \geq \VCdim(\mathcal{C}_{L,V})\geq 1$. 
    To show that $\Ldim(\mathcal{C}_{L,V}) \leq 1$, we provide an algorithm $A$ that makes no more than 1 mistake. $A$ labels all examples 0 until the adversary labels an example $(z,i)$ with 1. 
    As soon as this happens, $A$ makes a single mistake. However, the only concept in $\mathcal{C}_{L,V}$ consistent with this labeling is  $\Cert_z(x)$.  $A$ runs $V(z,w)$ with all possible $w\in\zo^{O(p(n))}$ to determine the lexicographically first $w^*$ such $V(z,w^*)=1$. With knowledge of $w^*$, it is trivial to determine the labeling of all future examples. Therefore, $A$ makes no further mistakes. 
\end{proof}

\begin{claim}[\Cref{claim:large-samples-upper-bound} for Online Learning]
\label{claim:many-mistakes-upper-bound}
  Let $\mcC$ be a concept class over $\zo^n$ such that every function $f\in\mcC$ is $p(n)$-sparse. Then there is an algorithm that online-learns $\mcC$ with at most $p(n)$ mistakes and takes time at most $O(n\log(p(n)))$ in each round. 
\end{claim}
\begin{proof}
We provide an algorithm $A$ that achieves these efficiency guarantees. 
    $A$ works by keeping a sorted list $\ell$ of all examples that have been labeled 1. Upon receiving a new example $x$, $A$ checks if $x\in\ell$. If it is, $A$ returns the label 1, if not it returns the label 0. If $A$ errs, learning that $x$ is labeled 1, it will update $\ell$ by inserting $x$ in the relevant position.  

We now analyze $A$'s runtime and mistake bound.
Each round, $A$ must search and potentially insert into a sorted list. Using a standard method like binary search, this takes time $O(n\log(p(n)))$. 

    Note that every time $A$ makes a mistake, the length of $\ell$ increases by 1. Therefore, $A$ makes $\leq |\ell|$ mistakes. Since $\ell$ can only contain examples labeled with 1, by assumption, $|\ell|\leq p(n)$. $A$ therefore achieves the desired mistake bound. 
\end{proof}

By definition, each function $\Cert_z\in \mathcal{C}_{L,V}$ is $cp(n)$-sparse. Since $\mathcal{C}_{L,V}$ has examples over $\zo^{n+\log (c p(n))}$, \Cref{claim:many-mistakes-upper-bound} yields the following corollary.

\begin{corollary}[\Cref{cor:large-samples-upper-bound} for Online Learning]\label{cor:many-mistakes-upper-bound}
Let $L\in \NTIME(t(n),p(n))$ and let $V$ be a time-$O(t(n))$ verifier for $L$ that uses size-$O(p(n))$ certificates. Then, there is an algorithm that online-learns $\mathcal{C}_{L,V}$  with at most $O(p(n))$ mistakes and takes time at most $O(n\log(p(n))+ \log^2(p(n))$ in each round. 
\end{corollary}

    \begin{fact}[Online learners imply PAC learners~\cite{Lit89}]\label{lem:online lower bound}  Let $\mcC$ be a concept class that is online-learnable by an algorithm that makes at most $m(n)$ mistakes and takes time at most $t(n)$ in each round. Then, there exists an algorithm that uses $M(n)\coloneqq O\left(\frac{1}\eps (m(n) +\log(\frac{1}{\delta}))\right)$ samples and $T(n)\coloneqq O( t(n)M(n)\log(M(n)))$ time and PAC-learns $\mcC$ to error $\eps$ with probability $1-\delta$. 
    \end{fact}

\subsubsection{Putting things together: Proof of~\Cref{thm:online}}

We use the same concept class $\mcC^\star\coloneqq \mcC_{L^\star,V^\star}$ as in the proof of~\Cref{thm:NP formal}. We show that it fulfills the three requirements of~\Cref{thm:online}.
    \begin{enumerate}
        \item By \Cref{claim:littlestone-dimension}, $\Ldim(\mcC^\star)=1$ and by \Cref{claim:enumerable}, the class is $\NTIME(p(n)\log(p(n)),p(n))$-enumerable. An $\NTIME(p(n)\log(p(n)),p(n))$ oracle can be simulated deterministically in time $2^{O(p(n))}$ and therefore all concepts in $\mcC^\star$ can be enumerated in time $2^{O(p(n))}$. It follows that there is an online learning algorithm that makes $\leq 1$ mistake and takes time $\leq 2^{O(p(n))}$ in each round.   
        \item By~\Cref{lem:online lower bound}, if there is an online learner for $\mcC^\star$ that makes at most $m(n)$ mistakes and takes time at most $t(n)$ in each round, then there there is an $O(m(n))$-sample, time-$O(t(n)m(n)\log(m(n)))$ algorithm that PAC-learns $\mcC^\star$ to error $\eps$ for any constant $\eps>0$. Since \Cref{thm:NP formal} holds for the concept class $\mcC^\star$, we can apply~\Cref{thm:NP formal:few samples} to conclude that $\NTIME(p(n))\sse \RTIME\left(2^{O(m(n))}T(n)\log(T(n))\poly(p(n))\right)$ where $T(n)= t(n)m(n)\log(m(n))$ as desired. 
        \item By \Cref{cor:many-mistakes-upper-bound}, there is an algorithm that online-learns $\mathcal{C}^\star$  with at most $O(p(n))$ mistakes and takes time at most $O(n\log(p(n))+ \log^2(p(n))$ in each round.
    \end{enumerate}



\section*{Acknowledgements}

We thank the FOCS reviewers for helpful feedback and suggestions. 

The authors are supported by NSF awards 1942123, 2211237, 2224246, a Sloan Research Fellowship, and a Google Research Scholar Award. Guy is also supported by a Jane Street Graduate Research Fellowship; Caleb by a Jack Kent Cooke Graduate Scholarship; and Carmen by a Stanford Computer Science Distinguished Fellowship and an NSF GRFP.

 \bibliography{ref}
 \bibliographystyle{alpha}


\end{document}